\documentclass[a4paper,envcountsame]{llncs}

\usepackage{microtype}%if unwanted, comment out or use option "draft"
\usepackage{amsfonts}
\usepackage{amsmath, amssymb}
\usepackage{tikz}
\usetikzlibrary{arrows}
\usepackage{todonotes}
\usepackage{algorithm}
\usepackage{algpseudocode}

\usepackage{graphicx}
\usepackage{caption}
\usepackage{hyperref}
\usepackage{multirow} 
\usepackage{soul}
\usepackage{caption}

\title{Eliminating recursion from monadic datalog programs on trees}
\author{Filip Mazowiecki \and Joanna Ochremiak \and Adam Witkowski}

\institute{University of Warsaw}

\newcommand{\Aa}{\mathcal A}
\newcommand{\Bb}{\mathcal B}

\def \downsquigarrow{\rotatebox[origin=c]{270}{$\rightsquigarrow$}}
\newcommand{\Ppr}{\mathcal P}
\newcommand{\Qpr}{\mathcal Q}

\newcommand{\Lab}{\Sigma_{\Ppr}}

\newcommand{\da}{\downarrow}

 \newcommand{\boxcomment}[1]{
 \begin{center}\fbox{\parbox{8cm}{#1}}\end{center}}

\newcommand{\fm}[1]{\boxcomment{{\bf FM: } #1}}
\newcommand{\jo}[1]{\boxcomment{{\bf JO: }#1}}
\newcommand{\aw}[1]{\boxcomment{{\bf AW: }#1}}

\renewcommand{\phi}{\varphi}

\newcommand{\Alf}{{\Sigma_{0}}}

\renewcommand{\succ}{{\downarrow}}
\newcommand{\desc}{{\downarrow_{\scriptscriptstyle +}}}

\newcommand{\cq}{\mathsf{CQ}}
\newcommand{\ucq}{\mathsf{UCQ}}
\newcommand{\cqall}{\cq(\succ,\desc)}
\newcommand{\ucqall}{\ucq(\succ,\desc)}

\newcommand{\cqsucc}{\cq(\succ)}
\newcommand{\ucqsucc}{\ucq(\succ)}

\newcommand{\datasucc}{\mathsf{Datalog}(\succ)}
\newcommand{\dataall}{\mathsf{Datalog}(\succ, \desc)}
\newcommand{\ldataall}{\mathsf{L\mbox{-}Datalog}(\succ, \desc)}

\newcommand{\cdatasucc}{\mathsf{Datalog}(\succ)}
\newcommand{\cldatasucc}{\mathsf{L\mbox{-}Datalog}(\succ)}

\newcommand{\NLogSpace}{\textsc{NLogSpace}}

\newcommand{\NP}{\textsc{NPTime}}
\newcommand{\PTime}{\textsc{PTime}}
\newcommand{\PSpace}{\textsc{PSpace}}
\newcommand{\ExpTime}{\textsc{ExpTime}}
\newcommand{\ExpSpace}{\textsc{ExpSpace}}

\newcommand{\TwoExpTime}{2\textsc{-ExpTime}}

\newcommand{\NPTime}{\textsc{NPTime}}% 

\newcommand{\IN}{\mathit{in}}
\newcommand{\OUT}{\mathit{out}}
\newcommand{\lab}{\mathit{lab}}

\newcommand{\query}{\mathcal{P} \land \lnot \mathcal{Q}}

\spnewtheorem*{lemma*}{Lemma}{\bfseries}{\rmfamily}
\spnewtheorem*{theorem*}{Theorem}{\bfseries}{\rmfamily}
\spnewtheorem*{proposition*}{Proposition}{\bfseries}{\rmfamily}
   \renewcommand{\fm}[1]{}
   \renewcommand{\jo}[1]{}
   \renewcommand{\aw}[1]{}
\begin{document}
\maketitle

\begin{abstract}
We study the problem of eliminating recursion from monadic datalog programs on trees with an infinite set of labels. We show that the boundedness problem, i.e., determining whether a datalog program is equivalent to \emph{some} nonrecursive one is undecidable but the decidability is regained if the descendant relation is disallowed. Under similar restrictions we obtain decidability of the problem of equivalence to a \emph{given} nonrecursive program. We investigate the connection between these two problems in more detail.

\end{abstract}

\section{Introduction}

\jo{PONIŻEJ MOJA ALTERNATYWNA WERSJA INTRODUCTION:}

\jo{Moja ogólna filozofia jest taka, żeby napisać krótko i konkretnie. Jakie problemy badamy i dlaczego (bo dobrze wyeliminowac rekursję ;) ). Jakie rozne ograniczenia nakladamy na programy i dlaczego (bo w ogolnosci jest duzo nierozstrzygalnosci). Na jakich strukturach pracujemy i jakie sa nasze wyniki. Bez motania się w różne rodzaje data-trees, różne wyniki wcześniejsze itd. Bo to tylko zamota. I sprawi, że się będzie trudniej czytać.}

Among logics with fixpoint capabilities, one of the most prominent is
datalog, which augments unions of conjunctive
queries (positive existential first order formulae) with recursion. Datalog
originated as a declarative programming language, but later found many
applications in databases as a query language. The gain in expressive power does not, however, come for free. Compared to unions of conjunctive
queries, evaluating a datalog program is harder~\cite{Vardi82} and basic properties such as containment or equivalence become undecidable \cite{Shmueli93}.

Since the source of the difficulty in dealing with datalog programs is their recursive
nature, the first line of attack in trying to optimize such programs is to eliminate the
recursion. It is well-known that a nonrecursive datalog program can be rewritten as a union of conjunctive queries. The main focus of this paper is therefore the equivalence of recursive datalog programs to unions of conjunctive queries.
%Unfortunately, it is also undecidable if a given datalog program is equivalent to some nonrecursive datalog program. 
\begin{example}
The programs in this example work on databases that use binary predicates \emph{likes} and \emph{knows}, and a unary predicate \emph{trendy}. First, consider the following pair of datalog programs:

%\begin{minipage}{0.03\textwidth}
% $\Ppr_1$
%\end{minipage}
%\begin{minipage}{0.46\textwidth}
%{\footnotesize   \begin{align*}
%      buys(X, Y) &\leftarrow  likes(X,Y)\\ 
%      buys(X,Y) &\leftarrow trendy(X), buys(Z,Y)
%    \end{align*} }
% \end{minipage}
%\begin{minipage}{0.03\textwidth}
% $\Ppr_1'$
%\end{minipage}
%\begin{minipage}{0.46\textwidth}
%{\footnotesize   \begin{align*}
%      buys(X, Y) &\leftarrow  likes(X,Y)\\ 
%      buys(X,Y) &\leftarrow trendy(X), likes(Z,Y)
%\end{align*} }
%\end{minipage}

%\begin{minipage}{0.03\textwidth}
% $\Ppr_1$
%\end{minipage}
\begin{minipage}{0.46\textwidth}
{\footnotesize   \begin{align*}
      &\Ppr_1 \\
      &buys(X, Y) \leftarrow  likes(X,Y)\\ 
      &buys(X,Y) \leftarrow trendy(X), buys(Z,Y)
    \end{align*} }
 \end{minipage}
%\begin{minipage}{0.03\textwidth}
% $\Ppr_1'$
%\end{minipage}
\begin{minipage}{0.50\textwidth}
{\footnotesize   \begin{align*}
      &\Ppr_1' \\
      &buys(X,Y) \leftarrow  likes(X,Y)\\ 
      &buys(X,Y) \leftarrow trendy(X), likes(Z,Y)
\end{align*} }
\end{minipage} \\
The program $\Ppr_1$ is recursive because its second rule refers to the predicate \emph{buys}. It can be shown that $\Ppr_1$ is equivalent to the nonrecursive program $\Ppr_1'$. Consider, on the other hand, the following pair of programs:

%\begin{minipage}{0.03\textwidth}
% $\Ppr_2$
%\end{minipage}
\begin{minipage}{0.46\textwidth}
{\footnotesize   \begin{align*}
      &\Ppr_2 \\
      &buys(X,Y) \leftarrow  likes(X,Y)\\ 
      &buys(X,Y) \leftarrow knows(X,Z), buys(Z,Y)
    \end{align*} }
\end{minipage}
%\begin{minipage}{0.03\textwidth}
% $\Ppr_2'$
%\end{minipage}
\begin{minipage}{0.50\textwidth}
{\footnotesize   \begin{align*}
      &\Ppr_2' \\
      &buys(X, Y) \leftarrow  likes(X,Y)\\ 
      &buys(X,Y) \leftarrow knows(X, Z), likes(Z,Y)
\end{align*} }
\end{minipage}
It can be shown that $\Ppr_2$ is not equivalent to the nonrecursive program $\Ppr_2'$. Moreover, this program is not equivalent to any nonrecursive program.
\end{example}

The example above (taken from \cite{Naughton89}) presents two approaches to eliminating recursion from datalog programs. Either we want to determine for a given datalog program if it is equivalent to \emph{some} nonrecursive datalog program or decide whether a given datalog program is equivalent to a \emph{given} nonrecursive program. These problems bear some similarities but in general they are separate. The latter is decidable \cite{ChaudhuriV92}, while the former, called the \emph{boundedness} problem, is not~\cite{GaifmanMSV93,HillebrandKMV95}.
%The main difficulty in showing this decidability result comes from another problem -- determining whether a given datalog program is contained in a given nonrecursive program. The equivalence and containment problems for datalog problems reduce to each other (see e.g. \cite{ChaudhuriV92}). For this reason papers on datalog often deal only with the containment problem.

Negative results
for the full datalog fueled interest in its restrictions
\cite{BenediktBS12,Bonatti04,CalvaneseGV05}. Important restrictions
include \emph{monadic} programs, using only unary
predicates in the heads of rules; \emph{linear} programs, with at most one use of an
intensional predicate per rule;
and \emph{connected} programs, where within each rule all variables that are mentioned
are connected to each other. Throughout this paper only monadic datalog programs are considered. 
In \cite{CosmadakisGKV88} Cosmadakis \textit{et al.} show that for such programs the boundedness 
problem becomes decidable. Moreover, they use the same techniques to prove that the containment 
problem of two monadic datalog programs is decidable. These results suggest that under some 
additional assumptions the boundedness problem and the equivalence problem are more related.
%These results suggest that for monadic programs the boundedness problem and the equivalence problem are more related.

% \jo{Napisałam "additional assumptions" zamiast "for monadic programs, bo chyba nasze meta-twierdzenie ostatecznie działa nie tylko dla monadic. I sugerowanie, że to monadyczność jest tu kluczowa chyba się kupy nie trzyma.}

In this paper we study connected, monadic datalog programs restricted to tree-structured databases. Our models are finite trees whose nodes carry labels from an infinite alphabet that can be tested for equality. Over such structures the problem of equivalence to a given union of conjunctive queries is known to be undecidable~\cite{AbiteboulBMW13,FFA}. We show that the boundedness problem is also undecidable. In some cases, however, we regain decidability of both problems in the absence of the descendant relation.
% CZESC JOASKA, OD TEGO MIEJSCA ZMIENILEM INTRO (czyli niewiele) :)
On ranked trees we show that the equivalence and the boundedness problems become decidable (in \TwoExpTime{}). On unranked trees we prove that the equivalence of a linear program to a non-recursive one is \ExpSpace-complete. We finish with an analysis of the connection between the equivalence and the boundedness problems and show that under some assumptions they are equi-decidable.

%On unranked trees we show decidability of the equivalence problem when restricting to linear programs. The boundedness problem and the generalization of the equivalence problem to the non-linear case are left as open problems.
%ZASTANAWIAM SIE CZY NAPISAC ZE EQUIVALENCE JEST TYLKO POPRAWIENIEM WYNIKOW A BOUNDEDNESS NOWE, ALE POKI CO WYWALILEM
% The results for the equivalence problem refine the complexities of the previous work in~\cite{FFA}. However, to our knowledge, this is the first paper to study the complexities of the boundedness problem of datalog programs on data trees.
%TO TEZ NIE WIEM CHYBA WYSTARCZY TO CO JEST W ORGANIZATION
% We finish with an analysis of the connection between the equivalence and boundedness problems and show that under some assumptions they are equi-decidable.

% Over ranked trees (in particular, words) the boundedness problem is decidable for monadic connected datalog programs (in \ThreeExpTime{}). The complexity is high but it can be lowered by restricting to linear programs. 
% %Over unranked trees we obtain decidability if the programs are monadic, connected and linear (in \ThreeExpTime{}). 
% We also show that the equivalence of a monadic connected linear datalog program to a given union of conjunctive queries is decidable (in \ExpSpace).

\emph{Organization}. In Section~\ref{sec:preliminaries} we introduce datalog programs and some basic definitions.
In Section~\ref{sec:equivalence} we deal with the problem of equivalence to a given nonrecursive datalog program.
In Section~\ref{sec:boundedness} we analyze the boundedness problem.
Finally, in Section~\ref{sec:boundedness_vs_equivalence} we explore the connection between the two approaches to eliminating recursion from datalog programs and show that under some assumptions the arising decision problems are equi-decidable.
We conclude in Section~\ref{sec:conclusion} with possible directions for future research. Due to the page limit most of the proofs are moved to the appendix.

% \fm{uwagi:
% 
% 1. dodalem pojecie nonrecursive programs. To wynika z tego ze w przykladzie 1 trudno mowic o tym ze to jest ucq a tez nie chce pisac w intro ze to sa zgrubsza te same rzeczy. Mozna zrobic jedna z trzech rzeczy: a) zostawic to tak jak jest i olac, b) zamienic wszystko na UCQ, c) dodac w preliminaries komentarz na temat nonrecursive programs.
% 
% 2. Jak pisze o monadic programs na ogolnych strukturach to nie wiem czy to jest najlepsze miejsce na to bo pozniej pisze o czyms innym a pozniej znowu sie do tego odwoluje, ale na razie nie mam lepszego pomyslu
% 
% 3. Pisząc intro doszedłem do wniosku że może lepiej podzielić rozdział o boundedness na meta twierdzenie i boundedness oddzielnie (nawet jesli ten o meta tiwerdzeniu bedzie krotki). Ale to zastanówcie sie wy.
% 
% 5. @Joaska, praca w ktorej jest napisane ze jak mamy nierownosci etykiet w UCQ (a wlasciwie ona jest tylko o CQ) to zawieranie jest nierostrzygalne jest tu \cite{BjorklundMS08}, mozesz sie do niej jakos odwolac i ew. zerknac do niej (najwazniejsze, czyli tabelka wynikow jest na koncu :P)}

\section{Preliminaries}
\label{sec:preliminaries}
In this paper we work over finite trees labeled with letters
from an infinite alphabet $\Sigma$. The trees are unranked by default,
but we also work with ranked trees, in particular with words. %We write
% $\nodes_t$ for the set of nodes of a tree $t$ and $\lab_t : \nodes_t \to
% \Sigma$ for the function assigning labels to nodes.
% 
% \jo{Czy ta notacja: $\nodes_t$ oraz $\lab_t : \nodes_t \to
% \Sigma$ jest nam naprawdę potrzebna? Póki co nigdzie jej nie używamy.}
 We use the standard notation for axes:
$\succ, \desc$ 
%,\pred, \ances$ 
stand, respectively, for child and descendant relations. %, parent and ancestor relations.
We assume that each node has one label. A binary relation $\sim$ holds between nodes with
identical labels and there is a unary predicate
$a$ for each $a\in\Sigma$, holding for the nodes labeled with $a$.

% Many papers on tree-structured databases work with a slightly
% different data model, called ``data trees'', where each node has a
% label from a finite alphabet and a data value from an infinite data
% domain. In that model labels can be used explicitly in the formulae,
% but cannot be directly tested for equality, and data values can be
% tested for equality, but cannot be used explicitly (as
% constants). These two models are very similar, but not directly
% comparable for query languages with limited negation. However, we can
% quite easily incorporate additional finite alphabet to our setting,
% obtaining a generalization of the two settings, and the complexity
% results do not change. Forbidding the use of constants from the
% infinite alphabet may affect some of our lower bounds but it does not
% affect decidability. For the undecidable fragments of datalog we give
% an additional sketch of the proof for ``data trees'' in the appendix.
% 
% \fm{moze lepiej to jakos zmienic}

%\subsection{Datalog}
We begin with a brief description of the syntax and semantics of
datalog; for more details see~\cite{Datalogfoundations} or~\cite{Ceri1990}. 
A {\bf datalog program} $\mathcal{P}$ over a relational signature $S$
is a finite set of rules of the form $\mathit{head} \leftarrow
\mathit{body}\,,$ where $\mathit{head}$ is an atom over $S$ and
$\mathit{body}$ is a (possibly empty) conjunction of atoms over $S$
written as a comma-separated list. All variables in the body that are not used in the head are implicitly quantified existentially. The {\bf size of a rule} is the number of different variables that appear in it.

% 
% \jo{Ja bym tu zdefiniowała size of a rule, bo potem ta definicja jest od czapy: The {\bf size of a rule} is the number of different variables that appear in it.}

The relational symbols, or predicates, in $S$ fall into two
categories. {\bf Extensional} predicates are the ones explicitly
stored in the database; they are never used in the heads of rules.  In
our setting they come from $\{\succ, \desc, \sim \} \cup
\Sigma$. The alphabet $\Sigma$ is infinite, but the program ${\cal P}$ uses
only its finite subset which we denote by $\Sigma_{\cal P}$.  
{\bf Intensional} predicates, used both in the heads and bodies, are defined by the
rules.

The program is evaluated by generating all
atoms (over intensional predicates) that can be inferred from the
underlying structure (tree) by applying the rules repeatedly, to the
point of saturation. 
%and then taking atoms matching the head of the goal predicate.
Each inferred atom can be witnessed by a \textbf{proof tree}: an
atom inferred by a rule $r$ from intensional atoms $A_1, A_2, \dots,
A_n$ is witnessed by a proof tree with the root labeled by $r$, and $n$
children which are the roots of the proof trees for atoms $A_i$ (if $r$ has no
intensional predicates in its body then the root has no children). 

%Each inferred atom can be witnessed by a \textbf{proof tree}: an atom inferred by a rule $r$ from intensional atoms $A_1, A_2, \dots, A_n$ is witnessed by a proof tree whose root has label $r$, and its children are the roots of the proof trees for atoms $A_i$ (if $r$ has no intensional predicates in its body then the root has no children). 

There is a designated predicate called
the {\bf goal} of the program. We will often identify the goal predicate with the program, i.e., we write $\Ppr(X)$ if the goal predicate of the program $\Ppr$ holds on the node $X$. When evaluated in a given database $D$, the program $\cal P$ results in the unary relation $\Ppr(D) = \{X \in D \mid  \text{such that } \Ppr(X) \text{ holds}\}$. If ${\cal P}(D) \subseteq {\cal Q}(D)$ for every database $D$ then we say that the program $\cal P$ is {\bf contained} in the program $\cal Q$. If the containment holds both ways then the programs $\cal P$ and $\cal Q$ are {\bf equivalent}.

% \jo{Dodałam definicje equivalent program i containement.}

% \jo{Może tą boolean query wyrzucić stąd i zdefiniować gdzie indziej? Gdzie tego w ogóle używamy?}

%For a given database $D$ we define a set $$\Ppr(D) = \{X \in D \ \textbar  \text{ such that } \Ppr(X) \text{ holds}\}$$ and $\Ppr_{b}(D) \in \{0,1\}$ which is equal 1 iff $\Ppr(D)$ is not empty.
%and $\Ppr_{r}(D) = \{X \in D \ \textbar  \text{ such that } \Ppr(X) \text{ holds, and $X$ is the root of $D$}\}$
%MOZE JEDNAK TO NAPISAC GDZIE INDZIEJ

% Each inferred intensional atom can be witnessed by a \textbf{proof
%   tree}, which explains how the atom was inferred.  A proof
% tree is a finite tree over the alphabet ${\cal  P}$, i.e.,  it uses
% rules of  program ${\cal  P}$ as labels; the rule in the
% root has the goal predicate in the head. If a node $v$
% contains a rule $A \leftarrow B_1, B_2, \dots, B_m, C_1, C_2,  \dots,
% C_n$, where the $B_i$'s use extensional predicates and the $C_j$'s use intensional
% predicates, then $v$ has $n$ children $v_1, v_2,  \dots,
% v_n$, and the head of the rule in $v_j$ matches atom $C_j$. In particular, 
% leaves are labelled with rules without intensional predicates in their bodies.
%
% Proof trees correspond to evaluation strategies: if the body of a rule
% contains intensional atoms then the program must evaluate new rules
% for nodes that were chosen as witnesses for these atoms. In the proof
% tree, for each such atom we add a child rule with matching head. 

\begin{example} \label{ex:datalog}
The program below computes the nodes from which one can reach some label~$a$ along a path where each node has a child with
identical label and a descendant with label~$b$ (or has label $b$ itself).

\vspace{-0.8ex}

\noindent\begin{minipage}{0.56\textwidth}
{\footnotesize    \begin{align*}
      P(X) &\leftarrow  X \succ Y, P(Y), X \succ  Y',  X\sim Y',
      Q(X) & (p_1)\\ 
      P(X) &\leftarrow a(X) & (p_2)\\
      Q(X) &\leftarrow X \succ Y, Q(Y) & (q_1)\\
      Q(X) &\leftarrow b(X) & (q_2)
    \end{align*}
}
%    \caption{$\mathcal{P}$}
%    \label{fig:gull}
  \end{minipage} 
  \hfill
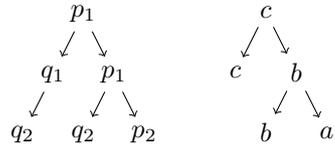
  \begin{minipage}{0.16\textwidth}
    \begin{tikzpicture}[scale=0.8]
      \node[draw=none] (z2) at (2, 5) {$p_1$};
      \node[draw=none] (M1) at (1.5, 4) {$q_1$} edge [<-] (z2);
      \node[draw=none] (M3) at (1, 3) {$q_2$} edge [<-] (M1);
      \node[draw=none] (M2) at (2.5, 4) {$p_1$} edge [<-] (z2);
      \node[draw=none] (x1) at (2, 3) {$q_2$} edge [<-] (M2);
      \node[draw=none] (y1) at (3, 3) {$p_2$} edge [<-] (M2);
    \end{tikzpicture}
%   \caption{$\Pi$}
%    \label{fig:tiger}
  \end{minipage}
  \hfill
  \begin{minipage}{0.12\textwidth}
    \begin{tikzpicture}[scale=0.8]
      \node[draw=none] (z2) at (2, 5) {$c$};
      \node[draw=none] (M1) at (1.5, 4) {$c$} edge [<-] (z2);
      \node[draw=none] (M2) at (2.5, 4) {$b$} edge [<-] (z2);
      \node[draw=none] (x1) at (2, 3) {$b$} edge [<-] (M2);
      \node[draw=none] (y1) at (3, 3) {$a$} edge [<-] (M2);
    \end{tikzpicture}
\end{minipage}

\vspace{-1.5ex}

\noindent The intensional predicates are $P$ and $Q$, and $P$ is the goal.
The proof tree shown in the center witnesses that $P$
holds in the root of the tree on the right.
\end{example}

% \jo{Tu jest jakaś nieścisłość. Piszemy wyżej o "goal rule" (i jest tu goal rule podkreślona). Ale napisane jest w przykładzie, że "goal is P". To w końcu w których wierzchołkach ten program jest spełniony: w tych, w których przyczepia się goal rule (ta podkreślona), czy w tych, w których zachodzi P? To jest zresztą ogólniejsze pytanie nie tylko dotyczące tego przykładu.}

The notion of proof trees comes from papers on datalog over general structures (see e.g.~\cite{ChaudhuriV92}). As shown in Example \ref{ex:datalog} proof trees illustrate how the program evaluates. While on general structures for a given proof tree one can always find a model such that the proof tree witnesses a correct evaluation of the program, on tree structures this is not so simple. One reason is that we allow only one label for every node. As a result, rules like $P(X) \leftarrow a(X), b(X)$ cannot be satisfied for $a \neq b$. Moreover, nodes have a unique father.
%Another property of trees is that nodes have a unique father. 
Because of this it is not easy to determine whether a given proof tree is a witness of an evaluation of the program on some model and it does not suffice to eliminate unsatisfiable rules. Proof trees for which such a model exists will be called {\bf satisfiable proof trees}.

%  \emph{rozumiem, że proof tree zawsze ma w liściach reguły nierekurencyjne? tzn. takie kawalki proof tree, ktore sa nie dokonczone, to nie sa proof tree?}
%TAK

\begin{example}
 \label{ex:sat_proof tree}
The program below goes down a tree along a path labeled with $a$. Then it goes up the tree until it finds a node labeled with $b$.
 
 \noindent\begin{minipage}{0.56\textwidth}
{\footnotesize    \begin{align*}
      P(X) &\leftarrow  X \succ Y, a(Y), P(Y)  & (p_3)\\ 
      P(X) &\leftarrow Q(X) & (p_4)\\
      Q(X) &\leftarrow Y \succ X, Q(Y) & (q_3)\\
      Q(X) &\leftarrow b(X) & (q_4)
    \end{align*}
}
%    \caption{$\mathcal{P}$}
%    \label{fig:gull}
  \end{minipage} 
%  \hfill
  \begin{minipage}{0.16\textwidth}
  {\footnotesize 
\begin{tikzpicture}[scale=0.65]
\node[draw=none] (q2) at (0, 10) {$p_3$};
\node[draw=none] (c2) at (0, 9) {$p_4$} edge [<-] (q2);
\node[draw=none] (r22) at (0, 8) {$q_3$} edge [<-] (c2);
\node[draw=none] (r12) at (0, 7) {$q_4$} edge [<-] (r22);
% \node[draw=none] (dl) at (0, 5) {$(\pi_1)$};
\end{tikzpicture}
}
% {\footnotesize \begin{align*}
%   p_3\rightarrow p_3 \rightarrow p_4 \rightarrow q_3 \rightarrow q_4 && (\pi_1) \\
%   p_3\rightarrow p_4 \rightarrow q_3 \rightarrow q_4 && (\pi_2) \\  
%  \end{align*}
%  }

%     \begin{tikzpicture}[scale=0.8]
%       \node[draw=none] (z2) at (5,2) {$p_1$};
%       \node[draw=none] (M1) at (4,2) {$p_1$} edge [<-] (z2);
%       \node[draw=none] (M3) at (3,2) {$p_1$} edge [<-] (M1);
%       \node[draw=none] (M2) at (2,2) {$p_2$} edge [<-] (z2);
%       \node[draw=none] (x1) at (1,2) {$q_1$} edge [<-] (M2);
%       \node[draw=none] (y1) at (0,2) {$q_2$} edge [<-] (M2);
%     \end{tikzpicture}
%   \caption{$\Pi$}
%    \label{fig:tiger}
  \end{minipage}
\begin{minipage}{0.12\textwidth}
  {\footnotesize 
\begin{tikzpicture}[scale=0.65]%MOZNA BARDZIEJ ZMNIEJSZYC STRZALKI ALE WYGLADA TROCHE ZLE WTEDY
\node[draw=none] (q2) at (0, 10) {$p_3$};
\node[draw=none] (c2) at (0, 9) {$p_3$} edge [<-] (q2);
\node[draw=none] (r22) at (0, 8) {$p_4$} edge [<-] (c2);
\node[draw=none] (r12) at (0, 7) {$q_3$} edge [<-] (r22);
\node[draw=none] (dl) at (0, 6) {$q_4$} edge [<-] (r12);
% \node[draw=none] (r13) at (0, 6) {$(\pi_2)$};
\end{tikzpicture}
}
\end{minipage}
  
%   \hfill
%   \begin{minipage}{0.12\textwidth}
%     \begin{tikzpicture}[scale=0.8]
%       \node[draw=none] (z2) at (2, 5) {$c$};
%       \node[draw=none] (M1) at (1.5, 4) {$c$} edge [<-] (z2);
%       \node[draw=none] (M2) at (2.5, 4) {$b$} edge [<-] (z2);
%       \node[draw=none] (x1) at (2, 3) {$b$} edge [<-] (M2);
%       \node[draw=none] (y1) at (3, 3) {$a$} edge [<-] (M2);
%     \end{tikzpicture}
% \end{minipage}
%\vspace{-1.8ex}

\noindent The first proof tree is satisfiable, but the second proof tree is not satisfiable because it enforces both labels $a$ and $b$ on the same node.
\end{example}

% \begin{wrapfigure}{1}{0pt}
% %\begin{center}
% \subfloat[$\Pi$]
% {
% \begin{tikzpicture}[scale=0.8]
% \node[draw=none] (z2) at (2, 5) {$p_1$};
% \node[draw=none] (M1) at (1.5, 4) {$q_1$} edge [<-] (z2);
% \node[draw=none] (M3) at (1, 3) {$q_2$} edge [<-] (M1);
% \node[draw=none] (M2) at (2.5, 4) {$p_1$} edge [<-] (z2);
% \node[draw=none] (x1) at (2, 3) {$q_2$} edge [<-] (M2);
% \node[draw=none] (y1) at (3, 3) {$p_2$} edge [<-] (M2);

% \end{tikzpicture}
% }
% \hspace{1cm}
% \subfloat[$t$]
% {
% \begin{tikzpicture}[scale=0.8]
% \node[draw=none] (z2) at (2, 5) {$c$};
% \node[draw=none] (M1) at (1.5, 4) {$c$} edge [<-] (z2);
% \node[draw=none] (M2) at (2.5, 4) {$b$} edge [<-] (z2);
% \node[draw=none] (x1) at (2, 3) {$b$} edge [<-] (M2);
% \node[draw=none] (y1) at (3, 3) {$a$} edge [<-] (M2);
% \end{tikzpicture}
% }
% %\end{center}
% \caption{A proof tree and its model}
% \label{fig:body}
% \end{wrapfigure}

% \jo{Ja bym w tym przykładzie zrobiła normalne proof tree z normalnymi strzałkami w dół. Bo to brzmi jakby to była jakaś konwencja na całą pracę.}

In this paper we consider only \textbf{monadic} programs, i.e.,
programs whose intensional predicates are at most 
unary. Moreover, throughout
the paper we assume that the programs do not use 0-ary intensional
predicates. 
% (except possibly for the goal predicate, but then it cannot be used
% in the body of any rule).
For general programs this is merely for the sake of
simplicity: one can always turn a 0-ary
predicate $Q$ to a unary predicate $Q(X)$ by introducing a dummy
variable $X$.  For connected programs (described below)
this restriction matters.

For a datalog rule $r$, let $G_r$ be a graph whose vertices are the
variables used in $r$ and an edge is placed between $X$ and $Y$ if the
body of $r$ contains an atomic formula $X \succ Y$ or $X \desc Y$. In $G_r$ we distinguish a {\bf head node} and {\bf intensional nodes}. The latter are all variables from the body of $r$ used by intensional predicates. A
program ${\cal P}$ is {\bf connected} if for each rule $r\in {\cal
P}$, the graph $G_r$ is connected\footnote{One could consider a definition allowing additionally nodes connected by the equality relation but we expect that this would be as hard as the disconnected case e.g. the main problem we leave open in Section~\ref{sec:equivalence}, the equivalence of child-only non-linear programs, becomes undecidable by the results of~\cite{FFA} for boolean queries.}.
% We say that ${\cal P}$ is {\bf downward} if
% for each rule $r\in {\cal P}$, $G_r$ is a directed tree whose root is
% the variable used in the head of $r$. 
% (except for rules for  0-ary goal predicates, where the root can be
% any variable). 
%In fact, each downward program is connected. 

% \fmm{Napisac gdzies, ze w kombinacjach booleowskich przyjmujemy ze
%   programy sa booleowskie}
% %\fmm{W algorytmach uzywamy booleowskich programow, tzn. 0-arnych.}

%%%%%%%% To powinno pojsc w skroconej formie do intro. 
%%%%%%%%
% 
% \jo{A może w paragrafie powyżej nie definiować co to jest linear? Wyrzucić ten fragment na razie (i definicję size of a rule, którą dałabym dużo wcześniej, a także przykład, że reguły mają jakiś size, bo to jest trywialne. Napisać paragraf tylko o tym, co to znaczy connected i zakończyć: The program from Example 1 is connected. A potem dać ten paragraf historyczny, który jest poniżej: Previous work on datalog bla bla bla.... . A jako kolejny, ten: We write Datalo(..) for the class of... . A datalog program is {\bf linear} if ... . For linear programs we shall use the letter L, e.g., .... . Mi się tak bardziej podoba, bo większy nacisk jest dzięki temu na pojęcie connected (nie mieszamy linear i connected w jednym paragrafie i nie wprowadzamy linear przed connected). No a connected to nasze glowne zalozenie jednak.}

Previous work on datalog on arbitrary structures often considered the
case of connected programs \cite{CosmadakisGKV88,GaifmanMSV93}. The
practical reason is that real-life programs tend to be connected (cf. \cite{BancilhonR86}). Also, rules which are not connected
combine pieces of unrelated data, corresponding to the $\emph{cross
product}$, an unnatural operation in the database context. It seems even
more natural to assume connectedness when working with tree-structured
databases. We shall do so.
We write $\dataall$ for the class of connected monadic datalog
programs, and $\datasucc$ for connected monadic programs that do not use
the relation $\desc$.

A datalog program is {\bf linear} if the right-hand side of each
rule contains at most one atom with an intensional predicate (proof
trees for such programs are single branches). 
For linear programs we
shall use the letter $\mathsf{L}$, e.g.,
$\mathsf{L\mbox{-}Datalog}(\succ)$ means linear programs
from $\datasucc$.
 The program from Example~\ref{ex:datalog}  is connected, but not linear.
The program from Example~\ref{ex:sat_proof tree} is both connected and linear.

{\bf Conjunctive queries} (CQs) are existential first order
formulae of the form $\exists x_1 \dots x_k \ \phi$, where $\phi$ is a
conjunction of atoms. We will consider {\bf unions of
conjunctive queries} (UCQs), corresponding to nonrecursive programs with a single
intensional predicate (goal) which is never used in the bodies of
rules. Since UCQs can be seen as datalog programs, we
can speak of connected UCQs and as for datalog, we shall always assume connectedness. We denote the classes of connected queries by $\cqall$, $\cqsucc$,
$\ucqall$, $\ucqsucc$, respectively.

\section{Equivalence}
\label{sec:equivalence}
\hyphenation{counter-example}

% Sometimes we will also consider a Boolean query $\Ppr_{Bool}(D)$ which equals 1 iff $\Ppr(D)$ is not empty.
% 
% \fm{equivalence = containment} 

% We briefly introduce some notation from \cite{FFA}.
% Each inferred atom can be witnessed by a \textbf{proof tree}: an
% atom inferred by the rule $r$ from intensional atoms $A_1, A_2, \dots,
% A_n$ is witnessed by a proof tree whose root has label $r$, and its
% children are the roots of proof trees for atoms $A_i$ (if $r$ has no
% intensional predicates in its body, the root has no children).
%  Notice that for linear programs a proof tree is actually a proof word.
 
% For a proof tree $T$ we define a canonical model. By Lemma \ref{} we can restrict to canonical models.

% and let $T(X)$ be a canonical model for $\Ppr(X)$.
% %The shortest path between two nodes in $T(X)$ is called the distance between two nodes.
% Note that the negative query can be matched only to nodes of distance at most $n$ from $X$.
% Because of this we introduce a new notion of a {\bf unary canonical model}.

% For every rule $r$ in $\Ppr$ and for every conjunct $C$ in $\Qpr$ we define its patterns $t_r$ and $t_C$. We sometimes identify these patterns with the set of their vertices.

% \fm{te nazwy na pewno do zmiany}

For datalog programs the containment problem can be reduced to the equivalence problem. Let $\cal P$ be a datalog program and let $\cal Q$ be a UCQ. Then $\cal P \subseteq \cal Q$ iff $\cal P \vee \cal Q \equiv \cal Q$. Notice that this reduction does not depend on the type of the programs (e.g., disallowing $\desc$ relation; or assuming linearity) but relies on the fact that datalog programs are closed under the disjunction.
 
The containment problem for datalog programs has been studied on trees in other contexts~\cite{AbiteboulBMW13,adamfilip,FrochauxGS14,FFA}. In~\cite{FFA} containment of datalog programs in UCQs on data trees was analyzed in detail for boolean queries, which are queries that return the answer 'yes' if they are satisfied in some node of a given database, and the answer 'no', otherwise. 
More formally, a datalog program $\Ppr$ defines a \emph{boolean query} $\Ppr_{Bool}(D)$ which equals $1$ iff $\Ppr(D)$ is nonempty and $0$ otherwise.
%More formally a boolean program $\Ppr$ results in the boolean relation $\Ppr_{Bool}(D)$ which equals $1$ iff $\Ppr(D)$ is nonempty. 

The containment problem is usually solved by considering the dual problem. For unary queries, it is the question whether there exist a database $D$ and $X \in D$ such that $\Ppr(X)$ and $\neg \Qpr(X)$, where $\neg \Qpr = D \setminus \Qpr(D)$. For boolean queries, it is the question if there exist a database $D$ and $X,Y \in D$ such that $\Ppr(X)$ and $\neg \Qpr(Y)$.
For datalog programs over trees, if we allow the $\desc$ relation this distinction does not make much of a difference (intuitively because using $\desc$ one can move from a node $X$ to any node $Y$). Thus a closer look at the proofs of Theorem~1 and Proposition~3 from \cite{FFA} gives the following.
%To understand the difference between unary and boolean queries let us look closer at the containment problem. Usually it is solved with the dual problem, i.e., whether there exists a database $D$ such that $\Ppr(D) \cap \neg \Qpr(D) \neq \emptyset$, where $\neg \Qpr = D \setminus \Qpr(D)$. In other words the question is if there exists a database $D$ and $X \in D$ such that $\Ppr(X)$ and $\neg \Qpr(X)$. A closer analysis of the boolean version shows that the containment problem is equivalent to the question if there exists a database $D$ and $X,Y \in D$ such that $\Ppr(X)$ and $\neg \Qpr(Y)$.
% 
%For programs that use $\desc$ relation this distinction does not make a big difference, intuitively because using $\desc$ one can move from the node $X$ to any node $Y$. Thus a closer look to the proofs of Theorem~1 and Proposition~3 from \cite{FFA} gives us the following.

\begin{proposition}
\label{thm:undec_containment}
Over ranked and unranked trees the containment problem of $\ldataall$ programs in $\ucqall$ is undecidable.
\end{proposition}
In the rest of this section we work only with fragments of datalog without the $\desc$ relation.
We start with ranked trees.

\begin{theorem}
\label{thm:ranked_trees}
%\label{thm:trees:twoexptime}
The containment problem is \TwoExpTime-complete for $\datasucc$ over ranked trees. In the special case of words it is \PSpace-complete.
\end{theorem}

The above result yields tight complexity bounds for the equivalence problem of $\datasucc$ programs to $\ucqsucc$ programs over ranked trees.
To prove Theorem~\ref{thm:ranked_trees} (see Appendices~\ref{subs:words} and~\ref{subs:trees}) we define automata that simulate the behavior of datalog programs, modifying the approach of \cite{FFA}. The new construction gives better complexity results for non-linear programs\footnote{In \cite{FFA} the non-linear case required an additional exponential blow-up. However, the improvement of complexity is not caused by considering unary instead of boolean queries. It is easy to see that Theorem~\ref{thm:ranked_trees} holds also in the boolean case.}.

In the rest of this section we focus on the equivalence problem of $\datasucc$ programs to $\ucqsucc$ programs over unranked trees. For the containment problem, this question was left open in~\cite{AbiteboulBMW13}.

For boolean queries, the containment problem of $\datasucc$ programs in $\ucqsucc$ programs was proved undecidable in~\cite{FFA}. Decidability was restored for the linear fragment, for which it was shown to be \TwoExpTime-complete. We improve the complexity for unary queries using different techniques (see Appendices~\ref{subs:upper} and~\ref{subs:lower}).

%Over ranked trees containment is \TwoExpTime-complete for $\datasucc$. In the special case of words containment is \PSpace-complete.
%\end{theorem}
%To prove this result we define automata that simulate the behavior of datalog programs, using a similar approach as in \cite{FFA}. This construction gives slightly better complexity results, because we get the same complexity for the linear case and non-linear case. In \cite{FFA} the non-linear case required an additional exponential blow-up. However, this improvement of the complexities does not come from changing the queries from boolean to unary, in fact it is easy to see that the lowered complexities carry on to the boolean queries.
%
%For the rest of this section we shall focus on the equivalence problem of $\datasucc$ programs to $\ucqsucc$ programs over unranked trees. This question for the containment problem was left as an open problem in~\cite{AbiteboulBMW13}.
%For boolean queries the containment problem of $\datasucc$ programs in $\ucqsucc$ programs was proved undecidable in~\cite{FFA}. Decidability was restored only for the linear fragment, with the complexity \TwoExpTime-complete. We improve this result for unary queries.

\begin{theorem}
\label{thm:containment:cl}
The containment problem of an $\cldatasucc$ program in a $\ucqsucc$ program is \ExpSpace-complete over unranked trees.
\end{theorem}

Unfortunately our approach does not generalize to the non-linear case. On the other hand, the proof of undecidability provided in~\cite{FFA} also cannot be adapted to work in our setting\footnote{Indeed, the main idea of the undecidability proof is to use the UCQ $\Qpr$ to find errors in the run of a Turing machine encoded by the program $\Ppr$. If the nonrecursive query $\Qpr$ is unary it can only find errors close to the node $X$, such that $\Ppr(X)$.}. We leave the question of the decidability of containment for non-linear programs as an open problem.

The following lemma is proved in Appendix~\ref{subs:lemma} (we do not assume linearity).
%We start with the ranked trees.

\begin{lemma}
\label{lemma:ucq_in_datalog}
The containment problem of $\ucqsucc$ queries in $\datasucc$ is in \NP{} over ranked and unranked trees.
\end{lemma}

As a corollary of Theorem~\ref{thm:containment:cl} and  Lemma \ref{lemma:ucq_in_datalog} we obtain the main result of this section. The lower bound is carried from the containment problem.

\begin{theorem}
 The equivalence problem of an $\cldatasucc$ program to a $\ucqsucc$ program is \ExpSpace-complete over unranked trees.
\end{theorem}

\section{Boundedness}
\label{sec:boundedness}
\newcommand{\pp}{{\mathcal A}^1_{\mathcal P}}
\newcommand{\ppp}{{\mathcal A}^2_{\mathcal P}}
 
Consider a datalog program $\cal P$ with a goal predicate $P$. By ${\cal P}^i(D)$ we denote the collection of facts about the predicate $P$ that can be deduced from a database $D$ by at most $i$ applications of the rules in $\cal P$. More formally, ${\cal P}^i(D)$ is the subset of ${\cal P}(D)$ derived using proof trees of height at most $i$, where the \emph{height} of a tree is the length of the longest path from its root to a leaf. Then obviously $${\cal P}(D) = \bigcup_{i \geq 0} {\cal P}^i(D).$$ We say that the program $\cal P$ is {\bf bounded} if there exists a number $n$, depending only on $\cal P$, such that for any database $D$, we have ${\cal P}(D) = {\cal P}^n(D)$. Intuitively this means that the depth of recursion is independent of the input database\footnote{Observe that we are only interested in the output on the goal predicate. This is why the property we consider is sometimes called the \emph{predicate} boundedness~\cite{HillebrandKMV95}.}.

Each proof tree corresponds to a conjunctive query in a natural way. Therefore, we can always translate a datalog program to an equivalent, but possibly infinite, union of conjunctive queries. If the program is bounded then it is equivalent to a finite subunion of its corresponding conjunctive queries. For full datalog it is known that the opposite implication is also true, i.e., a program is bounded iff it is equivalent to a (finite) UCQ~\cite{Naughton1991233}. The same holds for the class $\datasucc$:
%Moreover, if the program is connected then so is the corresponding union of conjunctive queries. For full datalog it is known that a program is bounded iff it is equivalent to a (finite) UCQ~\cite{Naughton1991233}.  In our case we obtain the following:

\begin{proposition}
\label{prop:bound_ucq}
Let ${\cal P \in \datasucc}$. Then $\cal P$ is bounded iff it is equivalent to a union of conjunctive queries ${\cal Q \in \ucqsucc}$.
\end{proposition}

We remark that the above characterization (which we prove in~Appendix~\ref{app:boundedness}) is based on the existence of so-called \emph{canonical databases} for CQs (see e.g.~\cite{Chandra:1977:OIC:800105.803397}) in $\datasucc$. The following example shows that without canonical databases equivalence to some UCQ does not necessarily imply boundedness. It relies on the fact that $\desc$ is the transitive closure of $\succ$. 
%By Proposition~\ref{prop:bound_ucq} there are no such examples for programs in $\datasucc$. 
%The class $\dataall$ does not have this property.

%\begin{example}
%The program $\cal P \in \dataall$ on the left below is not bounded. Its goal predicate is $P$. It goes down a tree looking for $a$ or $b$. It finds $a$ in one step. But finding $b$ can take arbitrarily long. The program $\cal P'$ on the right below is a UCQ equivalent to $\cal P$.
%
%\begin{minipage}{0.03\textwidth}
% $\Ppr$
%\end{minipage}
%\begin{minipage}{0.46\textwidth}
%{\footnotesize   \begin{align*}
%      P(X) &\leftarrow  X \desc Y, a(Y) \\
%      P(X) &\leftarrow  X \succ Y, Q(Y) \\
%      Q(X) &\leftarrow  X \succ Y, Q(Y) \\
%      Q(X) &\leftarrow  b(X)
%    \end{align*} }
%\end{minipage}
%\begin{minipage}{0.03\textwidth}
% $\Ppr'$
%\end{minipage}
%\begin{minipage}{0.50\textwidth}
%{\footnotesize   \begin{align*}
%     P(X) &\leftarrow  X \desc Y, a(Y) \\
%      P(X) &\leftarrow  X \desc Y, b(Y)
%\end{align*} }
%\end{minipage}
%\end{example}

\begin{example}
The program $\cal P \in \dataall$ on the left is not bounded -- finding $b$ in a tree can take arbitrarily long. The program $\cal P'$ on the right is a UCQ equivalent to $\cal P$.

\begin{minipage}{0.03\textwidth}
 $\Ppr$
\end{minipage}
\begin{minipage}{0.46\textwidth}
{\footnotesize   \begin{align*}
      P(X) &\leftarrow  X \desc Y, a(Y) \\
      P(X) &\leftarrow  X \succ Y, Q(Y) \\
      Q(X) &\leftarrow  X \succ Y, Q(Y) \\
      Q(X) &\leftarrow  b(X)
    \end{align*} }
\end{minipage}
\begin{minipage}{0.03\textwidth}
 $\Ppr'$
\end{minipage}
\begin{minipage}{0.50\textwidth}
{\footnotesize   \begin{align*}
     P(X) &\leftarrow  X \desc Y, a(Y) \\
      P(X) &\leftarrow  X \desc Y, b(Y)
\end{align*} }
\end{minipage}
\end{example}
We obtain a negative result for $\ldataall$ (see Appendix~\ref{app:undecidability}).

\begin{theorem}
\label{thm:boundedness:undec}
 The boundedness problem for $\ldataall$ is undecidable over words and ranked or unranked trees.
 \end{theorem}
%The proof via a reduction from the halting problem for Turing machines is provided in Appendix~\ref{app:undecidability}.
%Because of this result in the following we work only with the fragments of datalog without the $\desc$ relation. For decidability results we use the automaton-theoretic approach of~\cite{CosmadakisGKV88}.
In the following we work with fragments of datalog without the $\desc$ relation. For decidability results we use the automaton-theoretic approach of~\cite{CosmadakisGKV88}.

%For decidability results we use the automaton-theoretic approach of~\cite{CosmadakisGKV88}. We first look at the case of words, where many ideas can be illustrated without much of the technical difficulty of trees. 

\begin{theorem}\label{thm:boundedness:words:cl}
The boundedness problem for $\datasucc$ over words is in \PSpace{}.
\end{theorem}
In the case of trees the same technique can be applied but the complexity increases (see Appendix~\ref{app:ranked}).
\begin{theorem}
 \label{thm:boundedness:cl}
The boundedness problem for $\datasucc$ over ranked trees is in  \TwoExpTime{}. 
%For $\cldatasucc$ on ranked trees it is in \TwoExpTime{}. 
\end{theorem}
Over words, the relations $\succ$ and $\desc$ are interpreted as the ``next position'' and the ``following position''.
Let $X$ be a position in a word $w$. The $n$-\emph{neighbourhood} of $X$ in $w$ is an infix of $w$, which begins on position $\max(1,X-n)$ and ends on position $\min(|w|,X+n)$. The following lemma is motivated by Proposition 3.2 of~\cite{CosmadakisGKV88}. 
Its proof is provided in Appendix~\ref{app:ranked}.
%\begin{lemma}\label{lem:bounded}
%Let $\cal P$ be a $\cldatasucc$ program. Then the following conditions are equivalent: 
%\begin{enumerate}
%\item $\cal P$ is bounded,
%\item there exists $n > 0$ such that for every word $w$ and position $X$ if $X \in {\cal P}(w)$ then $X \in {\cal P}(v)$, where $v$ is the $n$-neighbourhood of $X$ in $w$. \label{prop}
%%\item there exists $n > 0$ such that for every word $w$ and position $X$ we have $X \in {\cal P}(w)$ iff $X \in {\cal P}(v)$, where $v$ is the $n$-neighbourhood of $X$ in $w$. 
%\end{enumerate}
%\end{lemma}
\begin{lemma}\label{lem:bounded}
Let $\cal P$ be a $\datasucc$ program. Then $\cal P$ is bounded iff
there exists $n > 0$ such that for every word $w$ and position $X$ if $X \in {\cal P}(w)$ then $X \in {\cal P}(v)$, where $v$ is the $n$-neighbourhood of $X$ in $w$.
\end{lemma}
\begin{proof}[of Theorem~\ref{thm:boundedness:words:cl}]
A word $w$ such that for some position $X$ in $w$ we have $X \in {\cal P}(w)$ but $X \not \in {\cal P}(v)$, where $v$ is the $n$-neighbourhood of $X$ in $w$ will be called an $n$-\emph{witness}. By Lemma~\ref{lem:bounded} a $\datasucc$ program $\cal P$ is unbounded iff there exist $n$-witnesses for arbitrarily big $n > 0$.

Consider a $\datasucc$ program $\cal{P}$. Let $\Sigma_0$ be an alphabet that contains the set of labels used explicitly in the rules of $\cal{P}$ together with $N$ ``fresh'' labels, where $N$ is the size of the biggest rule in $\cal{P}$. It is known~\cite{FFA} (and easy to verify) that any word $w$ can be relabeled so that the obtained word $w'$ uses only labels from $\Sigma_0$, and for each position $X$ we have that $X \in {\cal P}(w)$ iff $X \in {\cal P}(w')$. This is also true with respect to infixes, i.e., for every infix $v$ of $w$, and every position $X$ it holds that $X \in {\cal P}(v)$ iff $X \in {\cal P}(v')$, where $v'$ is the corresponding infix of $w'$. Hence, we can verify the existence of $n$-witnesses
%boundedness of $\cal P$ 
over the finite alphabet $\Sigma_0$.

In the proof of Theorem~\ref{thm:ranked_trees} (see Appendix~\ref{subs:words}) a nondeterministic automaton is introduced that recognizes words over the alphabet $\Sigma_0$ satisfying $\cal{P}$. More precisely, the constructed automaton $\cal{A_P}$ works over the alphabet $\Sigma_0 \times \{0,1\}$, and accepts a word $w$ iff it has exactly one position $X$ marked with $1$ such that $X \in {\cal P}(w)$. We denote the language recognized by $\cal{A_P}$ by $L(\cal{A_P})$. The size of this automaton is exponential in the size of $\cal{P}$. 
%Additionally, in the proof of Proposition 2 of \cite{FFA} a deterministic automaton was introduced that recognizes words over the alphabet $\Sigma_0$ satisfying $\cal{P}$. A slight modification of the given argument provides us with an automaton $\cal{A_P}$ over the alphabet $\Sigma_0 \times \{0,1\}$, such that the recognized words have exactly one position marked with $1$. A word $w$ is accepted by $\cal{A_P}$ if it has some position $X$ marked with 1 and $X \in {\cal P}(w)$. We denote the language recognized by $\cal{A_P}$ by $L(\cal{A_P})$. The size of this automaton is exponential in the size of $\cal{P}$. 
%

Similarly, we obtain an automaton $\cal N_P$ recognizing these words over the alphabet $\Sigma_0 \times \{0,1\}$ which have exactly one position marked with $1$ but do not belong to $L(\cal{A_P})$. The size of $\cal N_P$ is also exponential in the size of $\cal{P}$ (there is no exponential blow up because the constructions in Appendix~\ref{subs:words} go through alternating automata) and the language it recognizes will be denoted $L(\cal N_P)$. Note that this language is closed under infixes containing the marked position.

We define a nondeterministic automaton $\cal B_P$ which accepts exactly those words belonging to $L(\cal{A_P})$ which have an infix that belongs to $L(\cal{N_P})$. The states and transitions of $\cal B_P$ are the states and transitions of the product automaton $\cal A_P \times \cal N_P$ together with the states and transitions of two copies of the automaton $\cal A_P$ denoted $\pp$ and $\ppp$. Let $q_{init}$ be the initial state of $\cal N_P$. For each state $q$ of $\pp$ we add to $\cal B_P$ an epsilon transition from the state $q$ to the state $(q,q_{init})$ of the product automaton. Now, let $F$ be the set of final states of $\cal N_P$. For each state $q$ of $\ppp$ and each $q_{fin} \in F$ we add to $\cal B_P$ an epsilon transition from the state $(q,q_{fin})$ to $q$. 
%Additionally, we remove from $\cal A_P$ the transitions with labels from $\Sigma_0 \times \{1\}$. 
The initial state of $\cal B_P$ is the initial state of $\pp$ and the final states of $\cal B_P$ are the final states of $\ppp$. 
Hence, an accepting run of the automaton $\cal B_P$ starts in $\pp$, moves to the product automaton at some point, reads an infix that belongs to $L(\cal{N_P})$ and finally goes to $\ppp$ to accept.

Let $N$ be the number of states of the product automaton $\cal A_P \times \cal N_P$ plus $1$. Suppose that $\cal B_P$ accepts an $N$-witness $w$. Then, due to the pumping lemma, it accepts $n$-witnesses for arbitrarily big $n > 0$. To end the proof show that checking whether $\cal B_P$ accepts some $N$-witness is in \NLogSpace{} in the size of the automata $\cal A_P$ and $\cal N_P$ (i.e., in \PSpace{} in the size of $\cal P$).

An $N$-witness is a word that belongs to $L(\cal{A_P})$ but the $N$-neighbourhood of the position marked with $1$ belongs to $L(\cal{N_P})$. The \NLogSpace{} algorithm simulates a run of the automaton $\cal B_P$. The size of $\cal B_P$ is exponential in the size of ${\cal P}$ but its states and transitions
can be generated on the fly in polynomial space. The algorithm guesses a state from the $\cal A_P \times \cal N_P$ part and checks if it is reachable from the initial state. This is a simple reachability test which is in \NLogSpace{}. Then it guesses some run of the $\cal A_P \times \cal N_P$ part, counts the number of transitions done before the one marked with $1$, and ensures that it is at least~$N$. After the transition marked with $1$ it ensures that the automaton makes at least $N$ more transitions before leaving the $\cal A_P \times \cal N_P$ part. For both of these counting procedures we need $\log(N)$ tape cells. Finally, the algorithm performs a second reachability test to check if the automaton can reach a final state.

%Then the algorithm proceeds with simulating the run to end with an accepting state. This is a simple reachability test, which is in \NLogSpace{}.
% After the second epsilon transition the algorithm only needs to check if the final state is reachable which also can be done in logarithmic space.

%\fm{Ja bym napisał że zgadujemy stan w tej produktowej czesci i robimy reachability test do niego ze stanu poczatkowego. Potem te przejscia z licznikiem, a potem znowu reachability test do stanu koncowego (nie wydzielałbym momentu z epsilon przejsciem bo to i tak musi sie stac). Troche zmienilem zdanie koncowe. Poza tym trzeba dopisać standardową formułkę: ``The size of tu nazwa automatu is exponential in size of ${\cal P}$, but its  states and transitions can be generated on the fly in polynomial space.'' No bo oczywiscie nie mozemy tych automatow stworzyc w calosci bo sa za duze}
There are three possible ways of how an $N$-witness $v$ may look like. For simplicity, the algorithm described above does not deal with the case when the $N$-neighbourhood that belongs to $L({\cal N_P})$ is shorter then $2N+1$ (which can happen if it begins at the first position of $w$ or ends at the last position of $w$). Those possibilities can be verified similarly. \qed
\end{proof}
Notice that if $\Ppr$ is bounded then $N$ from the proof above is the bound on the depth of recursion. Since the size of the constructed automaton is exponential in the size of the program $\Ppr$, the UCQ which is equivalent to this program consists of proof trees of size at most exponential in the size of $\Ppr$.

\section{Boundedness vs equivalence}
\label{sec:boundedness_vs_equivalence}
In this section we focus on the similarities between the boundedness and the equivalence problem for datalog programs. 
In Sections \ref{sec:equivalence} and \ref{sec:boundedness} those problems are treated separately but with similar techniques. 
Also in \cite{CosmadakisGKV88}, where boundedness and equivalence are considered for monadic programs on arbitrary structures, 
both problems are solved using the same automata-theoretic construction. For these reasons we investigate the connection 
between the two problems in more detail. In contrast to the previous sections, in this section the structures under consideration are not necessarily trees or words.
\begin{definition}
A class $\cal C$ of datalog programs over a fixed class of databases is called {\bf well-behaved} if:
\begin{enumerate}
\item for every program $\cal P \in \cal C$ all the UCQs corresponding to the proof trees for $\cal P$ belong to $\cal C$,
\item containment of a UCQ in a datalog program is decidable for $\cal C$.
\end{enumerate}
\end{definition}
Condition (1) is satisfied for most natural classes of programs. In particular by the class of all datalog programs on arbitrary structures and the class $\datasucc$ on trees. For the class of datalog programs on arbitrary structures Condition (2) is also known to hold true (see \cite{ChandraLM81,CosmadakisK86,Sagiv88}). Lemma~\ref{lemma:ucq_in_datalog} shows that the class $\datasucc$ on trees satisfies Condition (2). Hence both those classes are well-behaved.
% We sketch the proof in Appendix~\ref{}.

% \fm{tu ta uwaga ze to nie dziala np jak mamy nierownosc}

% \fm{Nexptime na NP/coNP i w rozdziale o equivalence dopisac ze z tego i z containment wynika ograniczenie na equivalence}

We say that ${\cal C}$ has a {\bf computable bound}
if there exists a computable function $f$ such that if a datalog program ${\cal P}$ in ${\cal C}$ is bounded and $f({\cal P}) = n$ then ${\cal P}(D) = {\cal P}^n(D)$ for any database $D$, i.e., for bounded programs the function $f$ returns a bound on the depth of recursion. 
For programs which are not bounded $f$ returns some arbitrary natural numbers.
\begin{example}
\label{ex:full_bound}
Consider the full datalog. It follows from the results of~\cite{CosmadakisGKV88} that the class of monadic datalog programs on arbitrary structures has a computable bound. It is not stated explicitly but a closer analysis of the proofs gives that for a bounded program $\Ppr$ the depth of recursion can be bounded polynomially in the size of the automaton constructed to check if $\Ppr$ is bounded. For example, for a linear connected program the size of such an automaton is bounded exponentially in the size of the program.
\end{example}
The following theorem for a well-behaved class ${\cal C}$ with a computable bound establishes a connection between the problems of boundedness and equivalence to a given UCQ.
\begin{theorem}
\label{thm:boundedness:meta}
For any well-behaved class ${\cal C}$ with a computable bound the following conditions are equivalent: 
\begin{enumerate}
\item boundedness is decidable,
\item it is decidable whether two programs are equivalent, given that one of them is a UCQ.
\end{enumerate}
\end{theorem}
\begin{proof}
Let $f$ be the function from the definition of the computable bound. For the implication from (1) to (2), take programs $\cal P$ and $\cal Q$ which belong to $\cal C$ and assume that $\cal Q$ is a UCQ. Since $\cal C$ is well-behaved, we only need to show how to decide whether $\cal P$ is contained in $\cal Q$. It follows from the assumption that we can verify if $\cal P$ is bounded. If this is the case, then let $f({\cal P}) = n$. Observe that $\cal P$ is equivalent to the UCQ $\cal P'$ that corresponds to the proof trees for $\cal P$ of height at most $n$. It remains to decide whether the UCQ $\cal P'$ is contained in $\cal Q$.

Suppose now that $\cal P$ is not bounded and consider a union $\cal R$ of the programs $\cal P$ and $\cal Q$. More formally, let $\cal R$ be a program containing the rules of both programs $\cal P$ and $\cal Q$. If the predicate ${\cal Q}$ occurs in the program ${\cal P}$ we rename it so that the predicates do not coincide. The goal predicate $\cal R$ holds for $X$ iff we have ${\cal P}(X)$ or ${\cal Q}(X)$. For this we introduce two additional rules ${\cal R}(X) \leftarrow {\cal P}(X)$ and ${\cal R}(X) \leftarrow {\cal Q}(X)$. The atoms ${\cal Q}(X)$ are all inferred in one step. Therefore, if $\cal R$ is unbounded then there exists $X$ satisfying ${\cal P}(X)$ such that ${\cal Q}(X)$ does not hold, and hence $\cal P$ is not contained in $\cal Q$. If $\cal R$ is bounded then using $f$ we construct an equivalent UCQ $\cal R'$ and check whether it is equivalent to $\cal Q$. If this is the case then $\cal P$ is contained in $\cal Q$. Otherwise it is not.

%If $\cal R$ is unbounded then obviously $\cal P$ is not contained in $\cal Q$. If $\cal R$ is bounded then using $f$ we construct an equivalent UCQ $\cal R'$ and check whether it is equivalent to $\cal Q$. If this is the case, then $\cal P$ is contained in $\cal Q$. Otherwise it is not.

For the other implication, consider a datalog program $\cal P \in C$ and let $f({\cal P}) = n$. Then $\cal P$ is bounded iff ${\cal P}(D) = {\cal P}^n(D)$ for any database $D$. Let $\cal Q$ be the UCQ that corresponds to the proof trees of $\cal P$ of height at most $n$. It suffices to decide whether the programs $\cal P$ and $\cal Q$ are equivalent. But this is decidable from the assumption that ${\cal C}$ is well-behaved. \qed
\end{proof}
While assuming that a class of programs is well-behaved is natural, the existence of a computable bound is a strong assumption. It is needed since an algorithm that solves the boundedness problem might not be constructive, meaning that we do not know how big the equivalent UCQ is. However, deciding if such a function exists is usually as hard as solving the boundedness problem. From Example~\ref{ex:full_bound} we know that for monadic programs on arbitrary structures there exist constructive algorithms for the boundedness problem, and hence we have a computable bound.
On the other hand, the undecidability results of the boundedness problem for datalog on arbitrary structures rely heavily on the fact that such a computable bound does not exist. In~\cite{GaifmanMSV93,HillebrandKMV95} the authors present reductions from the halting problem for 2-counter machines and Turing machines. If a datalog program is bounded then the size of the equivalent UCQ corresponds to the length of an accepting run of these machines, which of course cannot be bounded by a computable function. The results of our paper are, in this sense, similar: the positive results provide computable bounds whereas the negative results rely on the fact that such a function does not exist.
For these reasons we conjecture that for well-behaved classes of datalog programs
%like full datalog on arbitrary structures or $\datasucc$ on trees, 
the decidability of the boundedness problem is equivalent to the decidability of finding a computable bound. If this conjecture holds true then Theorem~\ref{thm:boundedness:meta} becomes an implication from (1) to (2) because the opposite implication is trivially satisfied.

%\begin{theorem}
%Downward programs on unranked trees correpsond to programs on general structures.
%\end{theorem}

\section{Conclusions}
\label{sec:conclusion}
The equivalence to a given nonrecursive program and the boundedness problem for $\dataall$ are undecidable. To regain decidability we considered programs that do not use the $\desc$ relation.
We showed that equivalence to a given UCQ over ranked trees is decidable, and over unranked trees it is decidable in the case of linear programs.
%For the equivalence to a UCQ of non-linear programs we have a partial result: it is decidable for programs that start evaluating in the root of the database tree.
We also showed the decidability of boundedness on words and ranked trees. 
%For unranked trees we have decidability only for linear programs.
%The equivalence to a UCQ over unranked trees is decidable in case of linear programs.
%For the equivalence to a UCQ of non-linear programs we have a partial result: the equivalence is decidable for programs that start evaluating from the root of the database tree.
In the most general case of non-linear $\datasucc$ programs over unranked trees we do not know if the two problems under consideration are decidable 
and we leave these questions as open problems.
%In the end of Section \ref{sec:equivalence} we introduced a different type of queries, i.e., datalog programs that start evaluating from the root of the database.  It seems that in that setting our techniques could work even for non-linear programs.

We also investigated the connection between the boundedness and the equivalence to a UCQ. 
We showed that these problems are equivalently decidable for classes of programs with a computable bound. 
We suspect, however, that the existence of a computable bound for a class of programs is equivalent to the decidability of the boundedness problem. 
We also leave this as an open problem.

%WYWALIŁBYM TO :P
% Another direction for further research is to consider the class of downward programs introduced in~\cite{FFA}. This is a class of datalog programs which can evaluate only downward in the tree. In~\cite{FFA} the containment problem was investigated for downward programs but in the context of boolean queries.

% \subparagraph*{Acknowledgements}
% 
% I want to thank \dots

%%
%% Bibliography
%%

%% Either use bibtex (recommended), but commented out in this sample

%\bibliography{dummybib}

%% .. or use bibitems explicitely

\bibliography{datalog}

\newpage

\appendix
\section{Definitions}
\subsection{Automata}
Throughout the paper all decidability results use automata constructions. 
We briefly recall the standard automata model for ranked trees here.

A (bottom-up) \textbf{tree automaton} ${\cal A} = \langle \Gamma, Q,
\delta, F \rangle$ on at most $R$-ary trees consists of a finite
alphabet $\Gamma$, 
a finite set of states $Q$, a set of accepting states $F\subseteq Q$,
and transition relation $\delta \subseteq \bigcup_{i=0} ^ n Q \times
\Gamma \times Q^i$. A run on a tree $t$ over $\Gamma$ is a
labeling $\rho$ of $t$ with elements of $Q$ consistent with the
transition relation, i.e., if $v$ has 
children $v_1, v_2, \dots, v_k$ with  $k\leq n$, then $(\rho(v), \lab_t(v), \rho(v_1), \dots,
\rho(v_k) )\in \delta$. In particular, if $v$ is a leaf we have $(q,a)
\in \delta$. Run $\rho$ is accepting if it assigns a state from $F$ to
the root. A tree is accepted by $\Aa$ if it admits an
accepting run. The language recognized by
${\cal A}$, denoted by $L(\Aa)$, is the set of all accepted trees. We
recall that testing emptiness of a tree automaton can be done in
$\PTime$, but complementation involves an exponential blow-up. 
For a special case, when the model is words testing emptiness is in \NLogSpace.

% To work on unranked trees we encode them as binary trees in the usual
% way.  An unranked tree $t$ is turned into a binary tree $t_b$ as
% follows: put an order on children for every node in $t$; copy the root
% from $t$ to $t_b$; for each node $v$ from $t$ and let $v_b$ be its
% copy in $t_b$, as the left child of $v_b$ we put a copy of the
% previous sibling of $v$, and as the right child we put a copy of the
% last child of $v$. Thus we can use the same automata for the encoded
% unranked trees.

As an intermediate automata model, closer to datalog than the
bottom-up automata, we shall use the two-way alternating automata
introduced in \cite{CosmadakisGKV88}.  A \textbf{two-way alternating
automaton} $\Aa = \langle \Gamma, Q, q_I , \delta \rangle$ consists of
an alphabet $\Gamma$, a finite set of states $Q$, an initial state
$q_I\in Q$, and a transition function
\[\delta \colon Q \times \Gamma \to \mathrm{BC}^+\big(Q\times \{ -1,
0, 1\}\big)\] describing actions of automaton $\Aa$ in state $q$ in a
node with label $a$ as a positive boolean combination of atomic
actions of the form $(p,d) \in Q\times \{ -1, 0, 1\}$.

A run  $\rho$ of $\Aa$ over tree $t$
is a tree labelled with pairs $(q,v)$, where $q$ is a state
of $\Aa$ and $v$ is a node of $t$, satisfying the following
conditions: the root of $\rho$ is labelled with the pair consisting of
$q_0$ and the root of
$t$, and 
if a node of $\rho$ with label $(q,v)$ has children with labels
$(q_1,v_1),\dots, (q_n,v_n)$, and $v$ has label $a$ in $t$, then there
exist $d_1, \dots, d_n\in \{ -1, 0, 1\} $ such that: 
\begin{itemize}
\item $v_i$ is a child of $v$ in $t$ for all $i$ such that
  $d_i=1$;
\item $v_i = v$ for all $i$ such that $d_i=0$; 
\item $v_i$ is the parent of $v$ in $t$ for all $i$ such that $d_i=-1$; and
\item boolean combination $\delta(q, a)$ evaluates to  $\textit{true}$ when
  atomic actions \linebreak $(q_1, d_1), \dots, (q_n,d_n)$ are  substituted by
  $\textit{true}$, and other atomic actions are substituted by $\textit{false}$.
\end{itemize}
Tree $t$ is accepted by automaton $\Aa$ if it admits a finite run. By
$L(\Aa)$ we denote the language recognized by $\Aa$; that is, the set of trees
accepted by $\Aa$.

According to the definition above, two-way
alternating automata only distinguish between going up, down, and
staying where they are. In a more general model, appropriate for
ordered ranked trees, one could also distinguish between going to the first child, the
second child, etc. Given that our datalog programs are not able to
make such distinction, this simplified definition suffices. 

The computation model of two-way alternating automata is very similar
to that of datalog programs, making them a perfect intermediate
formalism on the road to nondeterministic bottom-up automata. From
there one continues thanks to the following fact. 
\begin{proposition}[\cite{CosmadakisGKV88}]
%,Vardi98
\label{prop:two-way}
 Given a two-way alternating automaton $\Aa$ (interpreted over words
 or ranked trees), one can compute (in time polynomial in the size of
 the input and output) single-exponential nondeterministic bottom-up
 automata recognizing the language $L(\Aa)$ and its complement,
 respectively. 
\end{proposition}
Notice that complementing two-way alternating automata is not trivial because there can be infinite runs that are not accepting.

\subsection{Canonical models and homomorphisms}
Let $r$ be a satisfiable rule of a datalog program $\Ppr$. Recall from Section \ref{sec:preliminaries} that $G_r$ is a graph of nodes from $r$. A {\bf pattern} $\pi_r$ has the same nodes and edges as $G_r$ but the type of edge between nodes ($\succ$ or $\desc$) is distinguished. The nodes are labeled with variable names. If there is an extensional unary predicate, e.g. $a(X)$, specified by the rule then we replace the label $X$ with $a$. We simulate the relation $\sim$ by repeating variable labels. % Notice that as in the rules variables correspond to nodes, in patterns they correspond to labels. 

Since in our setting the relation $\desc$ is disallowed, we can always transform a satisfiable rule $r$ into an equivalent rule $r'$ such that $\pi_{r'}$ is a tree. This is because our models are trees and therefore nodes that have a common child can be merged into one node. 

\begin{example}\label{ex:tree_pattern}
The rule $P$ is transformed into its tree version $P'$. On the right there are patterns corresponding to these rules. The repeated occurrence of $X$ represents the relation $\sim$ in the patterns.

\noindent\begin{minipage}{0.56\textwidth}
{\footnotesize    \begin{align*}
      P(X) &\leftarrow  X \succ Y, Y\succ Z, T \succ  Z,  a(T), X \sim Z \\
      & \downsquigarrow \\
      P'(X) &\leftarrow  X \succ Y, Y\succ Z,  a(Y), X \sim Z
    \end{align*}
}
  \end{minipage} 
  \begin{minipage}{0.16\textwidth}
    \begin{tikzpicture}[scale=0.8]
      \node[draw=none] (z2) at (2, 5) {$X$};
      \node[draw=none] (M1) at (1.5, 4) {$Y$} edge [<-] (z2);
      \node[draw=none] (M3) at (2, 3) {$X$} edge [<-] (M1);
      \node[draw=none] (M2) at (2.5, 4) {$a$} edge [->] (M3);
      \node[draw=none] (M2) at (3.5, 4) {$\rightsquigarrow$};
    \end{tikzpicture}
%   \caption{$\Pi$}
%    \label{fig:tiger}
  \end{minipage}
  \begin{minipage}{0.12\textwidth}
    \begin{tikzpicture}[scale=0.8]
      \node[draw=none] (z2) at (2, 5) {$X$};
      \node[draw=none] (M1) at (2, 4) {$a$} edge [<-] (z2);
      \node[draw=none] (M3) at (2, 3) {$X$} edge [<-] (M1);
    \end{tikzpicture}
\end{minipage}
\end{example}

A {\bf homomorphism} from a pattern $\pi_r$ to a model tree $t$ is a function between nodes that preserves the extensional predicates. A proof tree is witnessing an evaluation of the program on a given model $t$ iff for all rules there is a homomorphism from their patterns to $t$ such that the intensional nodes are mapped to the same nodes as the head nodes in the following rules.
The connection between patterns and datalog is explained in more detail in \cite{FFA}.

From a satisfiable proof tree we obtain a {\bf canonical model}. First we change rules to patterns and merge head nodes with intensional nodes. Nodes labeled with variables are relabeled with fresh labels (preserving the equalities forced by~$\sim$). The obtained graph can be seen as a pattern of the proof tree. Then we turn it into a tree similarly as in Example~\ref{ex:tree_pattern}.

It is easy to see that it suffices to consider the containment problem only on canonical models. If there is a model $t$ for $\Ppr \wedge \neg \Qpr$ then there is a witnessing proof tree for $\Ppr$ on $t$. The canonical model corresponding to this proof tree is also a model for $\Ppr \wedge \neg \Qpr$.

\section{Equivalence}\label{app:equivalence}
% Throughout the paper we work with the Boolean variant of
% containment, i.e., satisfiability of ${\cal P} \land \lnot {\cal Q}$,
% where ${\cal P}$ and ${\cal Q}$ are treated as Boolean queries, with
% all variables quantified out existentially. Using well known methods
% one can reduce the unary variant to the Boolean one (not by simply rewriting queries, though). 
% 
% To warm up, we first look at the special case of words (i.e., unary trees), where complexities are different than for trees of higher arity, and where some ideas can be illustrated without much of the technical difficulty of trees of higher arity.  

We decide containment by constructing an automaton that is non-empty iff there is a counterexample to containment.
To do this, we mark a single node in a tree, and use the automaton to verify if the goal predicates of programs in question are satisfied in this node.
Formally, we extend the alphabet by taking its product with $\{ 0,1 \}$, 
and recognize models which have exactly one node marked with $1$. To obtain tight complexity bounds, we use two-way alternating automata.
The same technique was used in \cite{FrochauxGS14}.

%Formally, we will show that both in case of words and trees it is enough to verify containment over a finite alphabet $\Sigma_0$
%Then we will extend the alphabet to $\Sigma_0 \times \{ 0,1 \}$ 
%and work with those models that have exactly one position (node) with $1$ as the second component of the label. 

\subsection{Special case: words}\label{subs:words}

Over words, the relations $\downarrow$ and $\desc$ are interpreted as the ``next position'' and the ``following position''.  

% Over data words and data trees, containment of
% monadic datalog programs with descendant is undecidable -- even
% containment of linear monadic programs in unions of conjunctive
% queries \cite[Proposition 3.3]{AbiteboulBMW13}. As discussed in
% Section \ref{sec:preliminaries}, despite obvious similarities between
% our setting and the setting of data trees, neither lower bounds nor
% upper bounds carry over immediately. The reduction in
% \cite{AbiteboulBMW13} relies on the presence of finite alphabet and
% cannot be directly adapted to our setting, but with a little effort
% the use of finite alphabet can be eliminated. Once descendant is disallowed, we immediately regain decidability.  

% \begin{proposition}
% \label{prop:words}
% %\label{prop:undecidability-words}
% Containment over words is  \PSpace-complete for $\datasucc$.
% \end{proposition}
\begin{lemma}
 \label{lem:words:automata}
Let $\Ppr \in \datasucc$ and let $\Sigma_0$ be a finite alphabet. There exists a two-way alternating automaton 
that accepts all words over $\Sigma_0 \times \{0,1\}$ with exactly one position with label $(a,1)$ for some $a \in \Sigma_0$, 
such that $\Ppr$ holds in that position. The automaton can be constructed in time polynomial in $|\Ppr|$ and $|\Sigma_0|$.
\end{lemma}
\begin{proof}
 Let us fix a program ${\cal P} \in \datasucc$ and a finite alphabet
$\Sigma_0$. The alphabet is $\Sigma_0 \times \{0,1\}$ but most of the time the second component is ignored. 
Since we work over words (and consider only connected
programs) without loss of generality we can assume that each rule $r$ is
of the form 
\[ H(x_0) \leftarrow \bigwedge_{i=k}^{\ell-1} x_i \da x_{i+1} \land
\psi(x_{k}, x_{k+ 1}, \dots, x_{\ell}),\] where $k\leq 0
\leq \ell$ and $\psi(x_{k}, x_{k + 1}, \dots, x_{\ell})$ is a
conjunction of atoms over unary predicates and $\sim$; that is, it does not use $\da$. This
means that the pattern corresponding to the body of $r$ is a word.

In the automaton ${\cal A}_{\cal P} = \langle \Sigma_0, Q, q_0, \delta
\rangle$ we are about to define we allow transitions of a slightly
generalized form: the transition function $\delta$ assigns to each
state-letter pair a positive boolean combination of elements of
\[Q\times \{- N, -N+1, \dots, N\}\] for a fixed constant $N\in \mathbb{N}$, rather
than just $Q\times \{-1, 0, 1\}$. The semantics of this is the
natural one: $(q,k)$ means that the automaton moves by $k$ positions
(left or right depending on the sign of $k$) and changes state to
$q$. Each generalized automaton can be transformed to a standard one
at the cost of enlarging the state-space by the factor of $2N+1$. In
our case $N$ will be bounded by the maximal number of variables used in a rule of $\cal P$.

Let us describe the automaton $\Aa_{\cal P}$. The state-space $Q$ is \[\Sigma_0 \cup
{\cal P} \cup \{q_0\}\,;\] that is, it consists of the letters from
$\Sigma_0$, the rules of $\cal P$ and an additional initial state
$q_0$.  The transition relation $\delta$ is defined as follows. In the
initial state, regardless of the current letter,  we loop moving to
the right until we reach the position 
in the word where we start evaluating $\cal P$:
\begin{align*}
 \delta (q_0, (\_\, ,0)) = (q_0, +1) , \\
 \delta (q_0, (\_\, ,1)) = (r_{\textrm{goal}}, 0)\,, \\
\end{align*}
where $r_{\textrm{goal}}$ is the goal rule of ${\cal P}$. 
This is the only case when ${\cal A}_{\cal P}$ does not ignore the component $\{0,1\}$ in the alphabet. 
That is we require that there is 1 in the second component when the first goal rule is applied.  When we are
in state $r\in {\cal P}$, regardless of the current letter,  we check that the body of $r$ can 
be matched in the input word in such a way that $x_0$ is mapped
to the current position:
\[ \delta(r,\_\,) = \bigwedge _{ a(x_i) } (a,i) \; \land
\bigwedge _{x_i \sim x_j} \bigvee _{\;b \in \Sigma_0}  (b,i)
\land (b,j) \; \land \bigwedge_{R(x_i)} \bigvee_{\;r' \in {\cal P}_R}
(r',i)\,,\]
where $a(x_i)$, $x_i\sim x_j$, and $R(x_i)$ range respectively over labels, $\sim$,
and intensional  atoms of $r$, and ${\cal P}_R \subseteq
\cal P$ is the set of rules defining intensional predicate $R$.
In state $a\in\Sigma_0$  we simply check that the letter in the
current position is $a$:
\[ \delta(a,a) = \top \,, \quad \textrm{and} \quad \delta(a,b) = \bot \textrm{ for } b\neq a\,.\]
Checking correctness and the size bounds for $\Aa_{\cal P}$ poses no
difficulties. Taking a product of $\Aa_{\Ppr}$ with an automaton  (of size linear in $|\Sigma_0|$) 
that checks if there is exactly one position with label $(a,1)$ for some $a \in \Sigma_0$
gives the automaton from the statement.
\end{proof}

Now we can show the proof of Theorem \ref{thm:ranked_trees} for the case of words.
\begin{proof}[of Theorem \ref{thm:ranked_trees} (words)]
In Proposition~2 of~\cite{FFA} it is shown that over words it suffices to
check satisfiability of $\query$ over an alphabet $\Sigma_0$
of linear size. 
% \aw{Tutaj jest poprzednia wersja}
% We will show that for a given 
% program ${\cal P} \in \datasucc$ and given finite alphabet $\Sigma_0$
% one can construct in $\PTime$ a two-way alternating automaton
% ${\Aa}_{\cal P}$ recognizing words over $\Sigma_0$ that satisfy $\cal 
% P$. Since the queries are unary, we need an automaton for $\Ppr(x) \wedge \neg \Qpr(x)$.
% For this reason we use an extended alphabet $\Sigma_0 \times\{0,1\}$. The automaton ${\Aa}_{\cal P}$ ignores the second component $\{0,1\}$, except that it starts the evaluation of $\Ppr$ only in nodes marked with $1$ (this will become more clear in the construction). Finally, we modify ${\Aa}_{\cal P}$ by taking the product with an automaton ${\Aa_{1}}$ which recognizes words over $\Sigma_0 \times\{0,1\}$ that have exactly one position with $1$ in the second component.
For programs $\Ppr$ and $\cal Q$, let $\Aa_{\Ppr}$ and ${\cal A}_{\cal Q}$ be alternating two-way automata given by Lemma~\ref{lem:words:automata}. From automata ${\cal A}_{\cal P}$ and ${\cal A}_{\cal Q}$,
by Proposition~\ref{prop:two-way}, we obtain one-way non-deterministic
automata $\Bb_{\cal P}$ and $\Bb_{\lnot{\cal Q}}$ of exponential size
that recognize respectively the language $L(\Aa_{\cal P})$ and the
complement of $L(\Aa_{\cal Q})$. From this we easily get a product
automaton $\Bb_{\cal P \land \lnot Q}$ equivalent to  the query ${\cal  P}(x) \wedge \neg {\cal Q}(x)$. 
Indeed, it accepts all words over $\Sigma_0$ with exactly one position $x$ marked with $1$, such that ${\cal  P}(x) \wedge \neg {\cal Q}(x)$.
  
%Moreover, with the extension of the alphabet to $\Sigma_0 \times \{0,1\}$ this automaton recognizes the query ${\cal  P}(x) \wedge \neg {\cal Q}(x)$. By the product construction with $\Aa_1$ the node $x$ is the same in both queries.

The size of $\Bb_{\cal P \land \lnot Q}$
is exponential in the size of ${\cal P, Q}$, but its states and
transitions can be generated on the fly in polynomial space.  To check
emptiness of $\Bb_{\cal P \land \lnot Q}$ we make a simple
reachability test, which is in \NLogSpace{}. Altogether, this gives a
\PSpace{} algorithm. 
\qed
\end{proof}

\subsection{Ranked trees}\label{subs:trees}

The results for words can be lifted to ranked trees: complexities
are higher, but the general picture remains the same. 
%Once again, first we construct an alternating, two-way automaton
%that checks if a program is satisfied in a marked node of a tree.
%Then we apply Proposition~\ref{prop:two-way} and construct an 
%automaton that is non-empty iff there exists a counterexample to containment of two programs.

\begin{lemma} \label{lem:trees:automata}
Let ${\cal P}\in\datasucc$ be a program with rules of size at most
$n$ and let $\Alf$ be a finite alphabet.
There exists a two-way alternating automaton $\Aa_{\cal P}$ of 
size ${\cal O}(\|{\cal P}\| \cdot |\Alf|^n \cdot n )$ recognizing trees over
$\Sigma_0 \times \{0,1\}$ with exactly one node with label $(a,1)$ for some $a \in \Sigma_0$,  such that $\Ppr$ holds in that node.
\end{lemma}
\begin{proof}
Let us fix a program ${\cal P}\in\datasucc$ and a finite alphabet
$\Alf$. Given that we are only interested in trees over alphabet
$\Alf$, we can eliminate the use of $\sim$ from ${\cal P}$: if a rule
contains $x\sim y$ we replace this rule with $|\Alf|$ variants in which
$x\sim y$ is replaced with $a(x)\land a(y)$ for $a\in\Alf$. 
The size
of the program grows by a ${\cal O}(|\Alf|^n)$ factor; the size of the
rules grows only by a constant factor. 

Since we are working on trees we can further transform the program so
that the patterns corresponding to the rules of the program are trees (with $\IN$ and
$\OUT$ nodes positioned arbitrarily). Indeed, it can be done by unifying
variables $x$ and $y$ whenever the rule contains $x\da z$ and $y \da
z$ for some variable $z$, and removing rules containing atom $u\da u$,
or atoms $a(u)$ and $b(u)$ for some variable $u$ and distinct letters $a$ and
$b$ (see Example~\ref{ex:tree_pattern}). This modification does not increase the size of the program. 

Finally, we rewrite each rule into a set of rules of the form 
\[H(x_0) \leftarrow a(x_0) \land \bigwedge_{i=1}^\ell \mathrm{ax}_i(x_0,x_i) \land
\psi(x_0, x_1, \dots, x_\ell)\] where $a\in \Alf$,
$\mathrm{ax}_i(x_0,x_i)$ is either $x_0\da x_i$ or $x_i\da x_0$, and $\psi(
x_0, x_1, \dots, x_\ell)$ is a conjunction of
(monadic) intensional atoms. That is, one rule can only test the
label and some intensional predicates for the current node, and demand
existence of neighbours (children or parents) satisfying some intensional 
predicates. This modification introduces auxiliary intensional
predicates, but the size of the program icreases only by ${\cal O}(n)$ factor. 

The resulting program is essentially a two-way alternating automaton $\Bb$,
only given in a different syntax. 
The automaton from the statement is obtained by modifying the automaton $\Bb$ similarly as in the case of words.
%$\Aa_{\Ppr}$ is a product of an automaton of constant size, enforcing existence of exactly 
%one node with label $(a,1)$ for some $a \ in \Alf$ and an automaton that finds a node with label $(a,1)$ for some $a$, then starts $\Bb$ from this node.
\end{proof}

\begin{proof}[of Theorem \ref{thm:ranked_trees} (trees)] 
In Theorem~1 of~\cite{FFA} it is shown that for trees it suffices to verify containment over a finite alphabet $\Sigma_0$, although for trees $\Sigma_0$ is of exponential size.
% For a  given program ${\cal P}\in\datasucc$ with rules of size at most
% $n$ and a finite alphabet
% $\Alf$, we shall construct a two-way alternating automaton $\Aa_{\cal
%   P}$ of size ${\cal O}(\|{\cal P}\| \cdot |\Alf|^n \cdot n )$ recognizing trees over
% $\Sigma_0$ that satisfy ${\cal P}$.
Using Lemma~\ref{lem:trees:automata} and
Proposition~\ref{prop:two-way} we reduce the containment problem to the
emptiness problem for a nondeterministic tree automaton of a double
exponential size in $|\Ppr|$, and test emptiness with the standard $\PTime$
algorithm. 
\qed
\end{proof}
The lower bounds can be obtained by a straightforward modifications of the results in \cite{FFA}.

\subsection{Satisfiability on unranked trees}

\begin{proposition}
\label{prop:sat:cl}
The satisfiability problem for $\cldatasucc$ on unranked trees is in \ExpTime{}. %For $\cdatasucc$ on unranked trees it is in \TwoExpTime{}. 
\end{proposition}
Before proving this result let us introduce the notation.
\begin{definition}\label{def:universal_tree}
Let $\Sigma_0$ be a finite alphabet. A {\bf universal} $\Sigma_0$-tree is a 
full $|\Sigma_0|$-ary tree over $\Sigma_0$ such that every non-leaf node has a child with each label from $\Sigma_0$.
For $a \in \Sigma_0, n \in \mathbb{N}$, we will denote by $U^{a}_{n}$ a universal $\Sigma_0$-tree of height $n$ and with $a$ in the root. 
\end{definition}
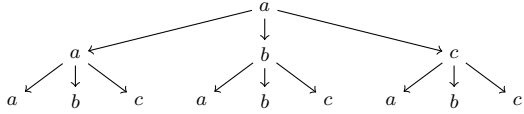
\begin{figure}[!h]
\begin{center}
 \begin{tikzpicture}[scale=0.83, transform shape]
 \node[draw=none] (a1) at (5,2.5) {$a$};
  \node[draw=none] (a2) at (2,1.75) {$a$} edge[<-] (a1);
  \node[draw=none] (b2) at (5,1.75) {$b$} edge[<-] (a1);
  \node[draw=none] (c2) at (8,1.75) {$c$} edge[<-] (a1);
  \node[draw=none] (a31) at (1,1) {$a$} edge[<-] (a2);
  \node[draw=none] (b31) at (2,1) {$b$} edge[<-] (a2);
  \node[draw=none] (c31) at (3,1) {$c$} edge[<-] (a2);
  \node[draw=none] (a32) at (4,1) {$a$} edge[<-] (b2);
  \node[draw=none] (b32) at (5,1) {$b$} edge[<-] (b2);
  \node[draw=none] (c32) at (6,1) {$c$} edge[<-] (b2);
  \node[draw=none] (a33) at (7,1) {$a$} edge[<-] (c2);
  \node[draw=none] (b33) at (8,1) {$b$} edge[<-] (c2);
  \node[draw=none] (c33) at (9,1) {$c$} edge[<-] (c2);
 \end{tikzpicture}
\caption{A universal $\Sigma_0$-tree $U^{a}_{2}$ for $\Sigma_0 = \{a,b,c\}$ } 
\end{center}
% %\includegraphics{universaltree}
\end{figure}

The proof will proceed as follows. First, we will show that if $\cal P$ is satisfiable, then it is satisfiable in a universal $\Sigma_0$-tree.
Then it is easy to see (combining Lemma \ref{lem:trees:automata} and Proposition \ref{prop:two-way}) that the set of universal $\Sigma_0$-trees satisfying $\Ppr$ is regular and recognized by an automaton with number of states double exponential in $|\Ppr|$. 
For linear programs however, we can do better and get an \ExpTime{} algorithm. 
\begin{lemma}\label{lemma:universal}
Let $\cal P \in \cdatasucc$ and let $\Sigma_0$ be a finite set of labels, s.t. $\Sigma_{\cal P} \subseteq \Sigma_0$.
The program $\Ppr$ is satisfiable iff $\Ppr$ is satisfiable in a universal $\Sigma_0$-tree.
\end{lemma}
\begin{proof}
It suffices to show that if $\Ppr$ is satisfiable, then it is satisfied in some universal $\Sigma_0$-tree. 
The other direction is obvious.
Let $t$ be a model for $\Ppr$. Recall that $\Sigma_{\Ppr}$ is the set of constants used in $\cal P$.
First, we can change all labels from $t$ that are not in $\Sigma_{0}$ to a single label chosen from $\Sigma_{0}$ (preserving the equalities).
Since $\datasucc$ programs do not use negation and $\not \sim$ this operation can only make the set $\Ppr(t)$ bigger. 
Next, we perform the following operation. If a node of $v$ has two or more children with the same labels 
then we merge these children into one node $v$. The resulting node has children from both of the merged nodes. 
It is easy to check that this operation preserves homomorphisms and does not change the emptiness of the set $\Ppr(t)$.
We apply this procedure until there are no siblings with the same label.
Finally we add nodes to the obtained tree so that it becomes a universal $\Sigma_0$-tree. 
Of course adding nodes cannot decrease the set $\Ppr(t)$, which finishes the proof. \qed
\end{proof}
% In this section, we will identify datalog rules with their tree patterns, and write about 
% homomorphisms from rules meaning homomorphisms from the corresponding patterns.
From now on we assume that $\Ppr$ is a linear program.
We will actually prove a stronger result that will be useful for deciding the containment of a datalog program in a UCQ.
We will show an algorithm for calculating all possible ways of evaluating the program $\Ppr$ in the universal $\Sigma_0$-tree such that the evaluation uses the root of this tree. 

First, we need to introduce a notion of a partial matching of a datalog program.
We say that a rule is matched to a tree $t$ if there is a homomorphism from its pattern into $t$.
Let $r_1r_2 \dots r_n$ be a proof word.
A {\bf partial matching} $m$ of a program $\cal P$ into a tree $t$ 
is an infix $r_i\dots r_j$ of a proof word such that all the rules $r_{i+1},\dots,r_{j-1}$
are matched completely and $r_i$ and $r_j$ are partially matched, such that the images of the intensional nodes are equal to the following head nodes. 
% For proof trees instead of an infix we take  a connected fragment of the proof tree and require partial homomorphisms from the patterns of the rules in the root and leaves; and full homomorphism from the patterns of the remaining rules.
%  A subtree of a proof tree is matched if all
% its rules are matched and additionally all those homomorphism agre, 
% that is if a rule $r$ is a child of a rule $r'$ in the proof subtree, 
% then the head node of $r$ and the node from $r'$ that 
% has the intensional predicate $r$ are mapped to the same node in the tree.
% \aw{.... to jest fatalny fragment, dwie zasady, ktore maja node'y, i drzewo, ktore tez je ma....} 

% We want to calculate the set of partial matchings that start and end in the root of an universal tree.
Each partial matching $m$ can be represented by a pair of partial homomorphisms from the patterns of the first and the last rule of the infix of the proof word. 
% In the non-linear case, we represent a subtree $\tau_s$ of a proof tree by a pair of partial homomorphism from the pattern of the rule in the root of $\tau_s$
% and a set of partial homomorphisms from the pattern of the leaves of $\tau_s$.
We are interested in the partial matchings that map one of the nodes of the pattern to the root of the tree. Thus each partial homomorphism can be represented as a partial function from pattern $\pi$ into $\Sigma_0$.
% The $\mathsf{NON}$ symbol will mean that this node has not yet been matched.
Of course there are also partial matchings with nodes mapped below the root of the tree, and one end of a partial matching may be not possible to extend. This situation can arise when the goal rule is at the beginning of the matching; or the non-recursive rules are in the last position (leaves).
We use an additional symbol $\mathsf{OK}$ to mark this situation.

We denote the set of all partial matchings of $\Ppr$ by $\mathit{Match}({\Ppr})$.
% \aw{nie wiem, czy ten zbiór Match ma być explicite wypisany}
The size of $\mathit{Match}({\Ppr})$ is exponential in the size of $\Ppr$.
Obviously it suffices to calculate the set of all partial matchings into a tree to determine if it satisfies $\cal P$.
\begin{lemma}\label{lem:partialmatchings}
Let $\Sigma_0$ be a finite alphabet. 
The set of partial matchings of $\Ppr$ matched in a root of any $\Sigma_0$-universal 
tree can be calculated in time exponential in $|\Ppr|$ for linear programs. 
%and double exponential for non-linear ones.
% The size of this set is at most exponential in $|\Ppr|$ for linear programs and double exponential for non-linear ones.
 \end{lemma}
% \fm{gdzie te zasady sa przymatchowane? nie wiem czy nie trzeba wprowadzic definicji patternu, bo nie ma homomorfizmow z zasad w datalogu}
% \jo{A co to w ogóle znaczy, że "... are matched" albo "... are matched completely"? Nie widzę definicji pojęcia "matched".}
% \jo{Partial homomorphisms from ... - do czego są te homomorfizmy?}
% \jo{"partial matchings in the root" - nie zrozumiałe. Rozumiem, że chodzi o takie, że program przyczepia się w korzeniu, ale nie ma definicji tego.}
\begin{proof}
For a tree $t$ we will denote the set of partial matchings in the root of $t$ by $\mathit{matched}(t)$.
Observe that because every partial matching of $\Ppr$ in $U^{a}_{n-1}$ is also 
a partial matching of $\Ppr$ in $U^{a}_{n}$ then $\mathit{matched}$ is monotonic, i.e., $\mathit{matched}(U^{a}_{n-1}) \subseteq \mathit{matched}(U^{a}_{n})$.
This observation yields a simple algorithm. 
There are $|\Sigma_0|$ different universal trees of height $n$. 
To calculate $\mathit{matched}(U^{a}_{n})$ for each $a \in \Sigma_0$ it suffices to join partial matchings 
from $\bigcup_{b \in \Sigma_0} \mathit{matched}(U^{b}_{n-1})$ using the root node labeled with $a$ and add the previously calculated $\mathit{matched}(U^{a}_{n-1})$.
Note that if $\mathit{matched}(U^{a}_{n}) = \mathit{matched}(U^{a}_{n+1})$, then $\mathit{matched}(U^{a}_{n}) = \mathit{matched}(U^{a}_{m})$ for all $m > n$.
Therefore, the described procedure requires at most $|\mathit{Match}({\cal P})|$ steps to terminate, each step takes
$O(|\mathit{Match}({\cal P})|)$ time which gives an \ExpTime{} algorithm. \qed
\end{proof}
\subsection{Proof of the upper bound in Theorem \ref{thm:containment:cl}}\label{subs:upper}
Let $\Ppr$ be a $\cldatasucc$ program and let $\Qpr$ be a $\ucqsucc$. Our goal is to determine whether for all databases $D$ we have $\Ppr(D) \subseteq \Qpr(D)$. We solve the dual problem and look for a counterexample for the containment, i.e., a database $D$ and a node $X \in D$ such that $X \in \Ppr(D)$ but $X \not \in \Qpr(D)$. Moreover, we can assume that $D$ is a canonical model. Let $n$ be the size of the biggest conjunct in $\Qpr$. Since $\Qpr$ is nonrecursive and connected, to determine if $X \in \Qpr(D)$ it suffices to check the subtree of $D$ containing nodes of distance at most $n$ from~$X$.

We shall refer to $\Ppr$ as the positive query and to $\Qpr$ as the negative query. We define an automaton $\Aa = (Q, A, \delta, q_0, F)$ that essentially recognizes satisfiable proof words for $\Ppr$  simultaneously checking if the negative query is satisfied on the canonical model of the read word.
The alphabet $A = \{r_1, \dots, r_m\}$ is the set of rules of the program $\Ppr$.
% As we shall see the automaton will not read the whole proof word, only the parts that are ``close'' to $X$. For the remaining parts, that are ``far'' from $X$, we need the states of the automaton obtained in Proposition \ref{}.

% \fm{trzeba gdzieś dopisać że rozpatrujemy tylko spójne programy i UCQ}
% \fm{dopisac co to jest satisfiable proof word}

%The intuition is that letters (rules) create a proof word for $\Ppr$.

% For rules such that their nodes in the canonical model are far from $X$ we use $\epsilon$-transitions.

% \fm{chyba trzeba jakoś lepiej napisać o tym co to znaczy że wierzchołki jakiejś zasady są daleko od $X$}

We define the set of states $Q$ as a cartesian product of three components, i.e., $Q = Q_1 \times Q_2 \times Q_3$. We describe each component separately. Recall that  $\Sigma_{\Ppr}$ denotes the set of constants used explicitly in rules of program $\Ppr$. Let $N$ be the size of the biggest rule in $\Ppr$. Let $B_1$ be an alphabet of $N$ different letters and let $B_2$ be an alphabet of $2n+1$ letters, disjoint from $B_1$.

In the first component $Q_1$ the automaton stores a labeled pattern corresponding to the currently read letter (rule). Formally, 
\[
Q_1 = \sum_{r \in A} (\Sigma_{\Ppr} \cup B_1 \cup B_2)^{\pi_r}.
\]
We identify the pattern $\pi_r$ with the set of its nodes, thus $Q_1$ is the set of patterns, whose nodes are labeled with elements of the set $\Sigma_{\Ppr} \cup B_1 \cup B_2$. The intended meaning of $B_1$ and $B_2$ will be explained later.

In the second component $Q_2$ %is the path starting from the current $\OUT$ node $v$ of nodes above that are of distance at most $2n$ from $X$. 
the automaton stores a word $w$ of length at most $2n+1$ and its position compared to the node $X$. Formally
\[
 Q_2 = \sum_{\substack{1 \leq i \leq 2n+1 \\ 0 \leq k,l \leq n}} \; (\Sigma_{\Ppr} \cup B_1 \cup B_2)^{[i]} \times (k,l),
\]
where $[i] = \{1, \dots, i\}$ and $(\Sigma_{\Ppr} \cup B_1 \cup B_2)^{[i]}$ is the set of words of length $i$ with labels from $\Sigma_{\Ppr} \cup B_1 \cup B_2$. This word is a representation of an ancestor-path starting from the intensional node $v_i$ of the current pattern stored in $Q_1$. This is necessary to verify if the proof word is satisfiable. The ancestor path could be arbitrary long but, as we will see, we only need to remember nodes that are of distance at most $n$ from $X$ (there are at most $2n+1$ such nodes). Additionally the automaton remembers how this path lays compared to $X$. For this it stores a pair of numbers $(k,l)$ such that $0 \leq k,l \leq n$. Let $v_a$ be the least common ancestor of $X$ and $v_i$. The number $k$ denotes the distance between $X$ and $v_a$, and the number $l$ denotes the distance between $v_a$ and $v_i$. Note that $k+l$ is the distance between $v_i$ and $X$. Also if $k = 0$ then $v_i$ is a descendant of $X$, and if $l = 0$ then $X$ is a descendant of $v_i$.

% \fm{current czy previous $\OUT$ node?, rysunek zdecydowanie}

The last component $Q_3$ is the set of partial homomorphisms of the patterns corresponding to CQs from the negative query $\Qpr$. Let $w$ be the word stored in the second component and let $A_\Qpr$  be the set of all conjuncts $\phi$ from $\Qpr$. Formally, $Q_3 = \sum_{\phi \in A_\Qpr} F_\phi$ where $F_\phi$ is the set of all partial functions from $\pi_\phi$ to $\Sigma_{\Ppr} \cup B_1 \cup B_2 \cup \{\flat\} \cup \{w_1,\dots,w_{|w|} \}$. The interpretation of the labels will be explained later.

% \fm{trzeba w preliminaries wymyślić jakąś dobrą notacje. Ja bym chciał, żeby reguły/CQ miały w sobie tego $X$, bo to tak powinno byc w tych unarnych zawieraniach}

We now define the transition relation $\delta$. Suppose that the automaton reads a new letter $r$. Let $q = (q_1, q_2, q_3) \in Q_1 \times Q_2 \times Q_3$ be the previous state. We show how the automaton calculates its new state $q' = (q_1', q_2', q_3')$.

In the first component the automaton starts from checking if the rule $r$ is proper for the intensional predicate in the previous rule; or if it is the first letter then the automaton checks if it is the goal predicate. If none of these cases holds then the automaton immediately rejects the word. Otherwise it labels $\pi_r$ in two phases. In the first phase it labels its head node $v_h$ with the same label that the intensional node in $q_1$ has. Also the labels of the nodes that are ancestors of $v_h$ must match the corresponding labels from the path in $q_2$. Then the automaton labels nodes that have an explicit label from $\Sigma_{\Ppr}$. In the second phase the automaton guesses the remaining labels from $\Sigma_{\Ppr} \cup B_1 \cup B_2$ respecting the $\sim$ relation. If there is a node on which $\sim$ forces two different labels, then the automaton rejects the word. This way we use a small alphabet to represent an arbitrary large set of labels. If in the state $q'$ we use a label that is also used in 
the state $q$ but $\sim$ does not force them to be the same then we assume that in the canonical model they are different labels.

% \fm{nie wiem czy dodawać stan śmietnik, czy lepiej pisać reject word, także notacja $\IN(t_r)$ do uzgodnienia}

In the second component the automaton updates first the pair $(k,l)$ so that it agrees with the location of the new intensional node with respect to $X$. Then it creates a new ancestor-path whose labels have to agree with the labels of the old path in $q_2$, and the labels of those nodes in $q_1'$ that are ancestors of $v_i$. The case when the distance of the new intensional node to $X$ is bigger than $n$ is explained later.% The remaining nodes get fresh labels from $B_1 \cup B_2$, i.e., different from the ones chosen in the first component. Notice that the size of $B_1$ and $B_2$ suffices to require fresh labels. 

In the last component the automaton starts from updating the old partial functions. All labels that appeared in $q_1$ but were not used in the first phase are replaced with $\flat$. The intended meaning is that these labels no longer appear in the model. Actually this is where we use the crucial feature of the canonical models. Since we use fresh labels whenever it is possible the automaton can forget all labels that will no longer appear.

% \fm{Może tę uwagę dopisać później też przy definicji kanonicznych modeli? Żeby zwrócić uwagę że dzięki niej możemy korzystać ze skończonego alfabetu (troche to śmieszne że używamy w modelu jak najwięcej etykiet, żeby w automacie było ich mało)}

The automaton forgets all partial homomorphisms that have unmapped nodes such that their label is forced by $\sim$ to be equal to a node labelled by $\flat$. This is because such homomorphisms can never be fulfilled. Then the automaton extends the remaining homomorphisms with new nodes from $q_1'$. The label $w_i$ denotes the fact that the node was mapped to the corresponding node from the path in $q_2'$. The next step is to relabel the partial homomorphisms so that they agree with the new path. This way the automaton knows where it can extend the homomorphisms. Note that if there is a partial homomorphism without any $w_i$ then it can be discarded because it cannot be extended. If at any time one of the homomorphisms becomes a full homomorphism then the automaton rejects the word.

% For rules that their intensional nodes are of distance bigger than $n$ from $X$ we use the satisfiability results.
% Intuitively all nodes that are of distance bigger than $n$ have no impact on the negative query so the automaton works only with the satisfiability problem of the positive query. The details are included in the appendix.
So far we explained the behavior for the letters in the proof word that have the intensional node of distance at most $n$ from $X$. This is of course not the only possible case, but we already noticed that nodes of bigger distance have no impact on the negative query. Because of this now we can use the results for the satisfiability problem. Suppose that the automaton reads a letter $r$ such that its intensional node is of distance bigger than $n$ from $X$.
The automaton updates the third component of its state in the usual way and rejects the word if a full homomorphism is found. Let $v$ be the ancestor of the intensional node in $r$ such that $v$ is of distance $n$ from $X$. The automaton assumes that there is a universal tree $t_v$ over the alphabet $\Sigma_{\Ppr} \cup B_1 \cup B_2$ (see Definition \ref{def:universal_tree}) below $v$. It calculates the set $\mathit{matched}(t_v)$ and finds all matchings that have the rule $r$ as the first rule with the node $v$ in the root. The automaton chooses one of the matchings but the last rule $r'$ can also have the intensional node below $v$. Then it proceeds with $r'$ as it did with $r$. Eventually the automaton guesses a matching such that the intensional node of the last rule $r''$ is of distance at most $n$ from $X$. Then it stores $r''$ in $q_1'$ and updates the other states in the usual way. If instead of the last rule there is $\mathsf{OK}$ then the automaton accepts the word.

Notice that the node $v$ could not exist. This happens when the least common ancestor of $X$ and the intensional node of $r$ is of distance bigger than $n$ from $X$. If such a situation occurs then, since we assumed that we work on canonical models, all nodes from the next rules will be of distance bigger than $n$ from $X$. Thus it suffices to check satisfiability starting from the rule $r$.

% Notice that we can restrict to the satisfiability problem. This is because the downward path is not unique and all new nodes in the canonical model will be of distance bigger than $n$.

We slightly modified the canonical models using universal trees. For the positive program we showed in Lemma \ref{lemma:universal} that we can use universal trees; and for the negative program we assured that the changes are on nodes that are of distance bigger than $n$ from $X$. 

The constructed automaton is non-empty iff there is a canonical model for $\Ppr \land \neg \Qpr$. We need to bound the size of the set of states. In the first component every labelled rule $(B_1 \cup B_2 \cup \Sigma_{\Ppr})^{\pi_r}$ is of exponential size in $|\Ppr|$ and the number of rules is bounded by the size of $\Ppr$. The second component is a set of triples: two numbers and a word of size at most $2n+1$, which is exponential in the size of $\Ppr, \Qpr$. The third component is the powerset of all partial homomorphisms which is double exponential in the size of $\Ppr$ and $\Qpr$. Thus the whole automaton is bounded double exponentially. However, its states and transitions can be generated on the fly in exponential space. To check its emptiness we make a simple reachability test, which is in \NLogSpace{}. We use the results about satisfiability to generate all transitions, but by Proposition \ref{prop:sat:cl} this can be done in \ExpTime. Altogether, this gives an algorithm in \ExpSpace.

% \section{Lower bounds}
\subsection{Proof of the lower bound in Theorem \ref{thm:containment:cl}}\label{subs:lower}
%\begin{proposition}
%\label{prop:expspace-hardness}
% The containment problem for $\cldatasucc$ is \ExpSpace-hard.
%\end{proposition}
%\begin{proof}
We consider the satisfiability problem of $\cal P
 \wedge \neg \cal Q$, where $\cal P \in \cldatasucc$ and $\cal Q \in \ucqsucc$. To prove hardness, for a number $n$ and a Turing machine $M$, we construct datalog programs $\cal P$ and $\cal Q$ of size polynomial in $|M|$ and $n$ such that $\cal P
 \wedge \neg \cal Q$ is satisfiable iff $M$ accepts the empty word using not more than $2^n$ tape cells. The program $\cal P$ will encode the run of the machine,
and the program $\cal Q$ will ensure its correctness.

Assume that $B$ is the tape alphabet of $M$, $Q$ is the set of states, $F$ is the set of accepting states and $\delta$ is the transition relation. The finite alphabet used by the programs will contain sets $B$ and $B \times Q$. The symbols from $B \times Q$ will be used to mark the position of the head on the tape and the state of the machine.
%The alphabet used by the program will contain $Q$ and two copies of $B$, 
%denoted $B$ and $\widehat{B}$.
%(symbols from $\widehat{\Sigma}$ are denoted with $\widehat{\phantom{x}}$)
%The symbols from $\widehat{B}$ will be used to mark the position of the head on the tape, 
%so $\widehat{a}$ means that the machine's head is over symbol $a$. We store the state of $M$ in the node below the node labeled with $\widehat{\phantom{x}}$.

%\fm{moze powinnismy w preliminaries wydzielic skonczona czesc alfabetu i napisac ze to jest wlasnie to, a pozniej np jak jest ze etykieta oznacza konfiguracje, to napisac ze ta etykieta jest z nieskonczonej czesci}

We now define the rules of the positive program $\cal P$. The program starts in a node labeled with $\top$. We encode each configuration of $M$ (the current state and the tape contents) by enforcing a full binary tree of hight $n$. For this we need the alphabet $\Sigma_{\cal P}$ to contain the set $\sum_{1\leq i \leq n} \{ (L,i), (R,i) \}$. The predicates $(L,i)$ and $(R,i)$ denote the left and right son of the previous node, respectively. The tape is encoded in the nodes below the leafs of the tree. The label of the node above the root of the tree is used as an identificator of the encoded configuration. We will refer to it as an \emph{identification node}.

%\fm{moze napisac ze L i R to left i right, a i to wysokosc w drzewie?}

The goal rule is 
\begin{align*}
G_{\cal P}(X) & \leftarrow \top(X), X \succ Y, Init(Y), Y \succ Z, conf(Z).
\end{align*}
It means that the encoding of the initial configuration of the machine, which identification node is labeled with $Init$, is stored in the tree (note that $Init$ belongs to $\Sigma_{\cal P}$).
The program will then traverse the configuration trees one by one in an infix order.
%\fm{tu bym dopisal ze init jest z tego skonczonego alfabetu}
\begin{align*}
conf(X) & \leftarrow  X\succ Y, (L,1)(Y), downleft^1(Y) \\
downleft^i (X) & \leftarrow  X\succ Y, (L,i+1)(Y), downleft^{i+1}(Y) \\
&i=1, \dots,n-1 \\
downleft^n(X) & \leftarrow X\succ Y, store(Y) \\
store(X) & \leftarrow a(X), Y \succ X, (L,n)(Y), upleft^n(X)\\
store(X) & \leftarrow a(X), Y \succ X, (R,n)(Y), upright^n(X)\\
&\text{for every symbol $a \in B \cup (B \times Q)$} \\
upleft^i(X) & \leftarrow Y\succ X, downright^{i-1}(Y) \\
&i=1,\dots,n \\
% ul^0_q(X) & \leftarrow  \\
downright^i (X) & \leftarrow X\succ Y, (R,i+1)(Y), downleft^{i+1}(Y) \\
& i=0,\dots,n-1 \\
% downright^n(X) & \leftarrow X\succ Y, store(Y) & \\
upright^i(X) & \leftarrow Y\succ X, (R,i-1)(Y), upright^{i-1}(Y) \\ 
&i=2,\dots,n \\
upright^i(X) & \leftarrow Y\succ X, (L,i-1)(Y), upleft^{i-1}(Y) \\ 
&i=2,\dots,n \\
upright^1(X) & \leftarrow Y\succ X, next(Y).
\end{align*}
Observe that when we reach $downleft^n$ we stop traversing the tree and the program uses the rule $store$ to write the content of the tape. That is why there is no rule $downright^n$.

The program finishes traversing the tree in $next$ and goes to the next configuration of the machine. We ensure that the identification node of the next configuration has the same label as the root of the tree which encodes the previous one. This will enable the negative program to check the correctness of the encoding.
\begin{align*}
next(X) & \leftarrow Y \succ X, Z \succ Y, Z \succ \hat Y, \hat Y \sim X, \hat Y \succ \hat X, conf(\hat X).
\end{align*}

We finish when we find an accepting state. That is for every letter $a \in B$ and every $q \in F$ we have two non-recursive rules
\begin{align*}
store(X) & \leftarrow (a, q)(X), Y \succ X,  (L,n)(Y) \\
store(X) & \leftarrow (a, q)(X), Y \succ X,  (R,n)(Y).
\end{align*}

Now let us define the rules of the negative program $\cal Q$, which will be a disjunction of queries describing possible errors in the encoding. The content of the tape has to be defined uniquely. Hence, for each pair of different symbols $a$ and $b$ from $B \cup (B \times Q)$ we have a rule
\begin{align*}
G_{\cal Q}(X) \leftarrow & \top(X), X \succ Y, Y \succ X_0, X_0 \succ X_1, X_1 \succ X_2, \ldots, X_n \succ  Z_1, X_n \succ Z_2, a(Z_1), b(Z_2).
\end{align*}
%Similarly, we have a rule ensuring that in each configuration tree there is at most  one cell of the tape marked as being the position of the head of the machine. That is, for each pair of different symbols $a$ and $b$ from $B \times Q$ there is a rule
%\begin{align*}
%Q(X) \leftarrow & \top(X), X \succ Y, Y \succ X_0, X_0 \succ X_1, X_0 \succ X'_1, \ldots, \\
%& \ldots, X_{n-1} \succ X_n, X'_{n-1} \succ X'_n, X_n \succ Z, X'_n \succ Z', a(Z), b(Z').
%\end{align*}

%\fm{ja bym rozwazyl jakas notacje w stylu $\succ_n$ i podobną na oznaczenie dwoch sciezek glebokosci $n$ o tych samych etykietach na sciezce}

We cannot ensure that each configuration tree has its identification node labeled differently, but we can guarantee that trees with the same labels of the identification nodes encode the same configurations. For each pair of different symbols $a$ and $b$ from $B \cup (B \times Q)$ we introduce a rule
\begin{align*}
G_{\cal Q}(X) \leftarrow & \top(X), X \succ Y, X \succ \hat Y, Y \sim \hat Y, Y \succ X_0, \hat Y \succ \hat X_0, X_0 \succ X_1,  \hat X_0 \succ \hat X_1, X_1 \sim \hat X_1, \ldots, \\
& \ldots, X_{n-1} \succ X_n, \hat X_{n-1} \succ \hat X_n, X_n \sim \hat X_n, X_n \succ Z, \hat X_n \succ \hat Z, a(Z), b(\hat Z).
\end{align*}
We can also easily enforce that the configuration tree labeled with $Init$ encodes the initial configuration of the machine with an empty word stored on the tape.

Finally we have to make sure that the way the positive program $\cal P$ moves from one configuration to another is consistent with the transition function of the machine. To do this we consider changes in the content of any three consecutive tape cells, i.e., we take all tuples $(a_1, a_2, a_3, b_1, b_2, b_3)$ of symbols from $B \cup (B \times Q)$, such that: if $a_1, a_2, a_3$ encode a content of three consecutive tape cells $i, i+1, i+2$, respectively, then it is not possible for the machine to have $b_1, b_2, b_3$ on those positions in the next configuration.
%\fm{ja bym dopisał zdanie o tym ze taka relacja ktora sprawdza czy te trojki sie zgadzaja wystarcza zeby sprawdzic poprawnosc przejsc, przy okazji punkt 1 chyba do wyrzucenia}
For each of those tuples there is a set of $2(n-1)$ rules in $\cal Q$. The rules are constructed depending on the least common ancestor of the three leafs which encode the consecutive tape cells. We write them down for $n=3$. There are two rules that deal with the case when the least common ancestor is the root of the tree
\begin{align*}
G_{\cal Q}(X) \leftarrow & \top(X), X \succ Y,  Y \succ X_0, X \succ \hat Y, X_0 \sim \hat Y, \\
& X_0 \succ X_1 \succ X_2 \succ X_3, X_0 \succ X'_1 \succ X'_2 \succ X'_3, X_0 \succ X''_1 \succ X''_2 \succ X''_3, \\
& (L,1)(X_1), (R,2)(X_2), (L,3)(X_3), \\ 
& (L,1)(X'_1), (R,2)(X'_2), (R,3)(X'_3), \\ 
& (R,1)(X''_1), (L,2)(X''_2), (L,3)(X''_3), \\
& Y \succ \hat X_0, \hat X_0 \succ \hat X_1 \succ \hat X_2 \succ \hat X_3, \hat X_0 \succ \hat X'_1 \succ \hat X'_2 \succ \hat X'_3, \hat X_0 \succ \hat X''_1 \succ \hat X''_2 \succ \hat X''_3, \\
& X_1 \sim \hat X_1, X'_1 \sim \hat X'_1, \ldots, X''_3 \sim \hat X''_3, \\
& X_3 \succ Z_1, a_1(Z_1), X'_3 \succ Z_2, a_2(Z_2), X''_3 \succ Z_3, a_3(Z_3), \\
& \hat X_3 \succ \hat Z_1, b_1(\hat Z_1), \hat X'_3 \succ \hat Z_2, b_2(Z_2), \hat X''_3 \succ \hat Z_3, b_3(\hat Z_3) \\
G_{\cal Q}(X) \leftarrow & \top(X), X \succ Y,  Y \succ X_0, X \succ \hat Y, X_0 \sim \hat Y, \\
& X_0 \succ X_1 \succ X_2 \succ X_3, X_0 \succ X'_1 \succ X'_2 \succ X'_3, X_0 \succ X''_1 \succ X''_2 \succ X''_3, \\
& (L,1)(X_1), (R,2)(X_2), (R,3)(X_3), \\ 
& (R,1)(X'_1), (L,2)(X'_2), (L,3)(X'_3), \\ 
& (R,1)(X''_1), (L,2)(X''_2), (R,3)(X''_3), \\
& Y \succ \hat X_0, \hat X_0 \succ \hat X_1 \succ \hat X_2 \succ \hat X_3, \hat X_0 \succ \hat X'_1 \succ \hat X'_2 \succ \hat X'_3, \hat X_0 \succ \hat X''_1 \succ \hat X''_2 \succ \hat X''_3, \\
& X_1 \sim \hat X_1, X'_1 \sim \hat X'_1, \ldots, X''_3 \sim \hat X''_3, \\
& X_3 \succ Z_1, a_1(Z_1), X'_3 \succ Z_2, a_2(Z_2), X''_3 \succ Z_3, a_3(Z_3), \\
& \hat X_3 \succ \hat Z_1, b_1(\hat Z_1), \hat X'_3 \succ \hat Z_2, b_2(Z_2), \hat X''_3 \succ \hat Z_3, b_3(\hat Z_3).
\end{align*}
And there are another two rules to deal with the case when the least common ancestor is labeled with $(L,1)$ or $(R,1)$
\begin{align*}
G_{\cal Q}(X) \leftarrow & \top(X), X \succ Y,  Y \succ X_0, X \succ \hat Y, X_0 \sim \hat Y, \\
& X_0 \succ X_1, X_1 \succ X_2 \succ X_3, X_1 \succ X'_2 \succ X'_3, X_1 \succ X''_2 \succ X''_3, \\
& (L,2)(X_2), (L,3)(X_3), \\ 
& (L,2)(X'_2), (R,3)(X'_3), \\ 
& (R,2)(X''_2), (L,3)(X''_3), \\
& Y \succ \hat X_0 \succ \hat X_1, \hat X_1 \succ \hat X_2 \succ \hat X_3, \hat X_1 \succ \hat X'_2 \succ \hat X'_3, \hat X_1 \succ \hat X''_2 \succ \hat X''_3, \\
& X_1 \sim \hat X_1, X_2 \sim \hat X_2, \ldots, X''_3 \sim \hat X''_3, \\
& X_3 \succ Z_1, a_1(Z_1), X'_3 \succ Z_2, a_2(Z_2), X''_3 \succ Z_3, a_3(Z_3), \\
& \hat X_3 \succ \hat Z_1, b_1(\hat Z_1), \hat X'_3 \succ \hat Z_2, b_2(Z_2), \hat X''_3 \succ \hat Z_3, b_3(\hat Z_3) \\
G_{\cal Q}(X) \leftarrow & \top(X), X \succ Y,  Y \succ X_0, X \succ \hat Y, X_0 \sim \hat Y, \\
& X_0 \succ X_1, X_1 \succ X_2 \succ X_3, X_1 \succ X'_2 \succ X'_3, X_1 \succ X''_2 \succ X''_3, \\
& (L,2)(X_2), (R,3)(X_3), \\ 
& (R,2)(X'_2), (L,3)(X'_3), \\ 
& (R,2)(X''_2), (R,3)(X''_3), \\
& Y \succ \hat X_0 \succ \hat X_1, \hat X_1 \succ \hat X_2 \succ \hat X_3, \hat X_1 \succ \hat X'_2 \succ \hat X'_3, \hat X_1 \succ \hat X''_2 \succ \hat X''_3, \\
& X_1 \sim \hat X_1, X_2 \sim \hat X_2, \ldots, X''_3 \sim \hat X''_3, \\
& X_3 \succ Z_1, a_1(Z_1), X'_3 \succ Z_2, a_2(Z_2), X''_3 \succ Z_3, a_3(Z_3), \\
& \hat X_3 \succ \hat Z_1, b_1(\hat Z_1), \hat X'_3 \succ \hat Z_2, b_2(Z_2), \hat X''_3 \succ \hat Z_3, b_3(\hat Z_3).
\end{align*}
%\jo{Napisać/narysować jaki kształt ma to kodowanie?}
%\jo{Nadal nie jestem pewna, czy muszę zapewniać, że drzewa z tym samym "identification node" wyglądają tak samo.}
%\jo{Te ostatenie zasady napisałam inaczej: dużo strzałek pod rząd. Może wcześniejsze też tak zmienić?}
\qed
%\end{proof}

\subsection{Proof of Lemma \ref{lemma:ucq_in_datalog}}\label{subs:lemma}
Take programs $\cal P \in \datasucc$ and $\cal Q \in \ucqsucc$. 
For every query $\phi$ in $\cal Q$ consider the pattern~$\pi_\phi$. 
Each of these patterns corresponds to a tree $t_{\phi}$ which is unique up to renaming of labels that are not explicitly mentioned by $\cal Q$.
Additionally, $t_\phi$ has one marked node $X$ corresponding to the head node of $\phi$.
It remains to check if ${\cal P}(X)$ holds for each of these trees. 
% The programs $\cal P$ and $\cal Q$ have a fixed number of intensional predicates and the size of databases under consideration is bounded (by the size of the biggest query in $\cal Q$). Therefore, the height of the proof trees we need to consider is bounded (polynomial in the size of $\cal P$ and $\cal Q$). The branching of the proof trees is also linearly bounded in the size of $\Ppr$.
It is well known that the combined complexity of monadic programs is \NPTime-complete. For each $t_{\phi}$ it suffices to guess the proof tree and verify the correctness of the guess.

\section{Boundedness}\label{app:boundedness}

\begin{proof}[of Proposition~\ref{prop:bound_ucq}]
The 'only if' part is obvious. For the 'if' part, suppose that a datalog program $\cal P$ is equivalent to a union of conjunctive queries ${\cal Q}$. For every rule $\phi$ of $\cal Q$ consider a pattern~$\pi_\phi$. With each of these patterns we associate a set of trees: the possible homomorphic images of $\pi_{\phi}$. Up to renaming of the labels which are not explicitly mentioned by $\cal Q$ there are finitely many such trees (this is because $\phi$ is connected and does not use the relation~$\desc$). We evaluate the program $\cal P$ on each of these trees and take $n$ to be the biggest number of applications of the rules in $\cal P$ that we need.
Now let $t$ be any tree. We will show that ${\cal P}(t) = {\cal P}^n(t)$. To this end, consider a node $X$ of $t$ such that ${\cal P}(X)$. Since the programs $\cal P$ and $\cal Q$ are equivalent, ${\cal Q}(X)$ also holds. This means that for some CQ $\phi$ of $\cal Q$ there is a witnessing homomorphism $h$ from $\pi_\phi$ to $t$. Thus, we need at most $n$ applications of the rules in $\cal P$ to derive ${\cal P}(X)$, because $h(\pi_\phi)$ is a fragment of $t$. \qed
\end{proof}

\subsection{Undecidability of the boundedness problem in general}\label{app:undecidability}
%In this section we prove Theorem~\ref{thm:boundedness:undec}.

\begin{proof}[of Theorem~\ref{thm:boundedness:undec}]
We will reduce the following problem:
given a  Turing machine $M$, are there arbitrary long runs of $M$ that start from an empty tape and end in the halting state (denoted HALT).
This problem is undecidable, because for a machine $M$, for every transition of $M$ that goes from state $q$ seeing symbol $a$ on tape to HALT state, 
we add another transition that stays in the state $q$ after reading $a$ and does not change the position of $M$'s head.
Thus, if $M$ had a run that halted, modified $M$ has arbitrary long halting runs.

Let $M$ be a Turing machine. We can assume without loss of generality that $M$ has one tape, semi-infinite to the right.
We will construct two programs, $P$ and $Q$. Program $P$ will find the encoding of the run of $M$ on an empty input in the tree and $Q$ will detect errors in the encoding.
The $Q$ program will be equivalent to a union of an UCQ.
Moreover, we will ensure that for every correct run of $M$, there is only one corresponding encoding.
Our program $P_{M}$ will be an alternative of $P$ and $Q$:
\begin{align*}
 P_M(X) :- P(X) \\
P_M(X) :- Q(X)
\end{align*}
If a tree contains an error in the encoding, $P_{M}$ will 
hold for every node of the tree in just 3 steps of the computation, because $Q$ qill be equivalent to an UCQ.
The constructed program will be \emph{not bounded} if and only if $M$ has arbitrary long halting runs.

The run of $M$ will be encoded as a word describing consecutive configurations. Configurations will be separated by \# symbols. The beginning of the encoding will be a $\mathsf{START}$ symbol and the end will be denoted by $\mathsf{END}$. Each position on the tape will be encoded by 4 consecutive nodes, $R-N-C-T$ where $R$ will denote row number, 
$N$ the number of the next row, $C$ the column number and $T$ the encoded tape symbol. 
$s$ will be marked with $0$ or $1$ denoting if the head of $M$ is in this position.
Because we consider trees, the encoding will be placed in the tree from some node upwards to the root. 
This way, the program will have only one path on which it can match. 
Otherwise (that is, going downwards in the tree) the correctness of the encoding cannot be guaranteed.

For each transition $\tau$ of $M$, there will be a set of rules verifying that the two consecutive encoded configurations of $M$ are 
consistent with $\tau$. Single rules will verify that the contents of the tape are copied/changed correctly between the configurations.
To ensure that, the rule will look at each 3 consecutive positions. 
For each triple of tape symbols, there will be rule that matches 3 positions encoding those tape symbols.
A rule $P_{\tau, a_1,a_2,a_3}^{i} (X)$ is true in $X$ if 3 positions described directly above $X$ contain symbols $a_1,a_2,a_3$ and 
the symbol in the next configuration in the same position as $a_2$ is also consistent with $\tau$. If the head of tape, this symbol should just be copied, 
but if head of $M$ is in the position with $a_1,a_2 or a_3$ the symbol can change between configurations.
The $i=1$ if the head of the tape was already seen in this configuration, $0$ otherwise.
For example, for a position where the head has not been seen in this configuration and there is no head in the inspected positions:
\begin{align*}
P_{\tau, (0,s_1), (0,s_2),(0,s_3)}^{0} (R_1) :- \\
&T_3 \succ C_3 \succ N_3 \succ R_3 \succ T_2  \succ C_2 \succ N_2 \succ R_2 \succ T_1 \succ C_1 \succ N_1 \succ R_1 \\   
&T_5 \succ C_5 \succ N_5 \succ R_5 \succ T_4 \succ C_4 \succ N_4 \succ R_4 \desc T_3 \\
&(0,s_1)(T_1), (0,s_2)(T_2), (0,s_3)(T_3), (0,s_2)(T_5) \\
&R_4 \sim R_5, R_4 \sim N_1, C_5 \sim C_2, C_4 \sim C_1 \\
&N_1 \sim N_2 \sim N_3, R_1 \sim R_2 \sim R_3, R_4 \sim R_5 \\
&P_{\tau, (0,s_2), (0,s_3), (i,s_4)}^{0} (R_1) 
\end{align*}
There will be such rule for any possible tape symbol $(i,s_4)$.
A quadruple $R_i, N_i, C_i, T_i$ of variables describes one position of the tape, in the configuration $R_i$, with next configuration $N_i$ and in column $C_i$.
The symbol stored in this position is $T_i$. 
Additionally, there will be  rules for changing rows, that checks two last positions before the \# and ensures that the next row is either the same length as the previous one or one position longer (that is, has 4 more nodes), depending on the movement of the head.
There will be also rules for the final row of the encoding (that is after reaching halting state), $P_{fin}$.
$P_{fin}$ will just go to the last \#, and $P$ will be true in the root of the tree (with $\mathsf{END}$ label) if $P_{fin}$ is matched in the last \#:
  
% For example, the rule for the changing rows when the head of the machine does not move right from the rightmost visited position of tape.
% \begin{align*}
%  P_{\tau, s} (X) :- Y_5 \succ Y_4 \succ Y_3 \succ Y_2 \succ Y_1 \desc X_4, \\
% X_8 \succ X_7 \succ X_6 \succ X_5 \succ X_4 \succ X_3 \succ X_2 \succ X_1 \succ  X, \\
% \#(X_4), \#(Y_5), \sim(X_2,Y_3),  
% \end{align*}
% Similar rule needs to be constructed for the case when the head moves form the rightmost seen position to the right and when the head of 
% $M$ is in this position.
The program $Q$ is given below, where $Q_{err}$ is an alternative of all possible errors in the encoding.
\begin{align}
Q(X) :- Y \desc X, Y \desc Z , Q_{err} (Z)  \\
Q(X) :- Q_{err}(X) \\
Q(X) :- X \desc Y, Q_{err}(Y)
\end{align}
Note the necessity of this triple alternative as $\desc$ is a proper descendant relation, that is $X \desc X$ does not hold.
This way, $Q$ holds in every node of the tree if $Q_{err}$ is found anywhere.
The possible errors are
\begin{enumerate}
\item \# or tape symbol appearing on the wrong position, for example detecting symbol $(0,s)$ used as a colum number
\begin{align*}
Q_{err}(X) :- \#(X), X_3 \succ X_2 \succ X_1 \succ X, (0,s)(X_3) \\
Q_{err}(X) :- \#(X), X_3 \succ X_2 \succ X_1 \desc Y \succ X, \sim (Y,X_1), (0,s)(X_3)
\end{align*}
Similarly such rules can be constructed for next row, row and \# used a tape symbol.
\item two consecutive \# symbols, detected by $Q_{err}(X) :- \#(X), X \succ Y, \#(Y)$. 
\item any node appears above the $\mathsf{END}$, detected by $Q_{err}(X) :- \mathsf{END}(X), Y \succ X$
\item any node appears below the $\mathsf{START}$, detected by $Q_{err}(X) :- \mathsf{START(X)}, X \succ Y$
\item row number used in two different rows, detected by 
\begin {equation*}
Q_{err}(X) :- \# (X), Y_1\succ X, Y_2 \desc Y_1, \# (Y_2), Z \succ Y_2, Z \sim Y_1 
\end {equation*}
\item the same column number twice in one row, detected by 
\begin{align*}
Q_{err}(X) :- \#(X), Z_3\succ Z_2 \succ Z_1 \desc Y_3 \succ Y_2 \succ Y_1 \succ X \\
Y_1 \sim Z_1, Y_3 \sim Z_3
\end{align*}
The last program works only if every row has distinct row number, which is ensured by previous rule.
\end{enumerate}
It is easy to see that $P_M$ is matched in every node of any tree that contains one of described errors, 
and in the root node of those databases that contain correct encoding of a halting run of $M$. Moreover, the computation of $P_M$ in those databases takes number of steps linearly proportional to the length of the encoding. Therefore, $P_M$ is unbounded if and only if $M$ has the arbitrary long halting run property. \qed
\end{proof}

\subsection{Boundedness on words and ranked trees}\label{app:ranked}
% 
% \jo{A MOŻE NAPISAĆ POD TYM TWIERDZENIEM KOMENTARZ, ŻE JAK WYNIKA Z PRZYKŁADU NR COŚ TAM (TEGO KTÓRY JA NAPISAŁAM POWYŻEJ), JEŚLI MAMY POTOMKA, TO BOUNDEDNESS TO NIE TO SAMO CO ROWNOWAŻNOŚĆ UCQ. WIĘC TAK NAPRAWDĘ ODPOWIADAMY TU NA PROCHĘ INNE PYTANIE... CZY COS W TYM STYLU. ZEBY SIE JAKOS WYTLUMACZYC Z TEJ SCIEMY KTÓRĄ PISZEMY W INTRODUCTION...}
% 
% \jo{Trzeba tak zrobić, żeby to twierdzenie nie miało dwóch osobnych numerów. W apendiksie ma też numer ;) }

\fm{Trzeba wytłumaczyc co to jest |t| i subtree to chyba nie jest dobra nazwa, raczej connected component czy jakos tak}

\begin{proof}[of Lemma~\ref{lem:bounded}]
One implication is immediate. If $\cal P$ is bounded then it is equivalent to a union of conjunctive 
queries $\cal Q$. The queries are connected so we can take $n$ to be the size of the biggest query in $\cal Q$.    

For the other implication, let us assume that $\cal P$ satisfies the condition: 
\begin{itemize}
\item there exists $n > 0$ such that for every word $w$ and position $X$ if $X \in {\cal P}(w)$ then $X \in {\cal P}(v)$, where $v$ is the $n$-neighbourhood of $X$ in $w$
\end{itemize}
with $n = n_0$. We will construct a union of conjunctive queries $\cal Q$ equivalent to $\cal P$. Recall that $\Sigma_{\cal P}$ denotes the set of labels that appear in the rules of program $\cal P$. Let us consider all words of length smaller or equal $2n_0 + 1$ and treat them as structures over the signature $\{\succ, \sim\} \cup \Sigma_{\cal P}$. These words have finitely many equality types. For each word $v$ that satisfies ${\cal P}$ we add to $\cal Q$ a query which defines the equality type of $v$.
It remains to show that $\cal P$ and $\cal Q$ are equivalent. The containment of $\cal Q$ in $\cal P$ is straightforward from the construction of $\cal Q$. Take a word $w$ and position $X$ such that $X \in {\cal P}(w)$. Then $X \in {\cal P}(v)$, where $v$ is the $n$-neighbourhood of $X$ in $w$. Since $v$ is a word of lenght at most $2n_0 +1$ it follows that $X \in {\cal Q}(v)$, and hence $X \in {\cal Q}(w)$. \qed
\end{proof}

We now move to the case of trees. 
%The same technique can be applied in this setting but the complexities increase.
%\begin{theorem}
% \label{thm:boundedness:cl}
%The boundedness problem for $\datasucc$ on ranked trees is in  \ThreeExpTime{}. For $\cldatasucc$ on ranked trees it is in \TwoExpTime{}. 
%\end{theorem}
First let us state the lemma equivalent to Lemma~\ref{lem:bounded} for ranked trees. 
For a tree $t$, the $n$-\emph{neighbourhood} of a node $X$ is a subtree of $t$ consisting of all nodes that are in distance at most $n$ from $X$.
\begin{lemma}\label{lem:trees:bounded}
Let $\cal P$ be a $\datasucc$ program over ranked trees. Then the following conditions are equivalent: 
\begin{enumerate}
\item $\Ppr$ is bounded,\label{wawa}
\item there exists $n >0$ such that for every tree $t$ and node $X$ if $X \in \Ppr(t)$ then 
$X \in {\cal P}(t')$, where $t'$ is the $n$-neighbourhood of $X$ in $t$.\label{wa}
\end{enumerate}
\end{lemma}
\begin{proof}
The proof is analogous to the proof of Lemma~\ref{lem:bounded}. Let $k$ be the rank of the considered trees. 
To show the implication from~\ref{wa} to~\ref{wawa} 
it is enough to notice that for given $n$ 
there are finitely many  equality types (with respect to $\Ppr$) of trees of height at most $2n+1$ (and thus, finitely many of equality types of $n$-neighbourhoods).
The equality type of each such $n$-neighbourhood is definable by a CQ, and a UCQ equivalent to $\Ppr$ 
is a union of those CQ's that are contained in $\Ppr$. \qed
\end{proof}

In the case of trees we define an $n$-\emph{witness} for $\Ppr$ to be a tree $t$ such that 
there exists a node $X$ in $t$ for which $X \in \Ppr(t)$ but $X \not\in \Ppr(t')$, where $t'$ is the $n$-neighbourhood of $X$ in $t$.
A \emph{witness} is a tree that is an $n$-witness for any $n > 0$.

\jo{W obu: proposition i corollary zmienilam $\cldatasucc$ na $\datasucc$, bo twierdzenie formulujemy dla $\datasucc$. Rozumiem, ze to sie zgadza?}

%As in the case of words, we have

\begin{corollary}
A $\datasucc$ program $\cal P$ over ranked trees is unbounded iff there exist $n$-witnesses for arbitrarily big $n > 0$.
\end{corollary}

We can now give the proof of Theorem~\ref{thm:boundedness:cl}. We restate it first.
\begin{theorem*}
The boundedness problem for $\datasucc$ over ranked trees is in  \TwoExpTime{}.  
\end{theorem*}
\begin{proof}
% \aw{moze po prostu trzeba przepisac ten lemat i dowod dac do App?}
\fm{trzeba poprawic argument o skonczonym alfabecie (na slowach tez) zgrubsza chodzi o to ze uzywajac malego alfabetu nie dodajemy nowych wyprowadzen }
To prove Theorem \ref{thm:boundedness:cl} we first show that boundedness can be verified over ranked trees over a finite alphabet.
\begin{lemma}
Let $\Ppr$ be a $\datasucc$ program. Then 
$\Ppr$ is bounded over ranked data trees with rank $R$ over $\Sigma$ iff $\Ppr$ is bounded over ranked trees with the same rank 
over a finite alphabet $\Sigma_0 \subseteq \Sigma$. The alphabet $\Sigma_0$ contains $\Lab$ and $|\Sigma_0 \setminus \Lab| \leq R^{|\Ppr|}$.
\end{lemma}
\begin{proof}
This proof is a slight modification of a proof from \cite{FFA}.
If $\Ppr$ is bounded over $\Sigma$ then it is clearly bounded over any finite subset of $\Sigma$.
Suppose that $\Ppr$ is bounded over $\Sigma_0$ but not bounded over $\Sigma$. Over $\Sigma_0$, $\Ppr$ is therefore equivalent 
to a UCQ $Q$ built of a finite number of proof words of $\Ppr$. Let $t$ be a tree over $\Sigma$ and $X$ a node in $t$ s.t. 
$X \in \Ppr(t)$ but $X \not\in Q(t)$. We will show that $t$ can be relabeled into a tree $t'$ over $\Sigma_0$ in a way
preserving any label comparison done by the rules of $\Ppr$.
Then, as $Q$ is a union of proof words of $\Ppr$, it must also hold that $X \in Q(t)$ iff $X \in Q(t')$, which is a contradiction since $\Ppr$ is not equivalent to $Q$ over ranked trees over $\Sigma$.

\aw{troche nadużywam tu notacji dla X, no ale nie wiem jak to lepiej napisac} 

\jo{moim zdaniem jest bardzo dobrze napisane :) }

Let $n$ be the size of the largest rule in $\Ppr$. Let $B \subseteq \Sigma
\setminus \Lab$ be a set of size $R^{|\Ppr|}$. We set $\Sigma_0 = B \cup \Lab$. 
We will describe a procedure that traverses the tree $t$ in a top-down
fashion, level by level, and changes the labels to elements of
$B$. This way the set of processed nodes consists of $i$
full levels starting from the root, and some nodes from the level $i+1$.

Let $v$ be a node on level $i+1$ -- the next one to process, and let $u$ be the node
$n-1$ edges up the tree (or the root if $v$ is too close to the
root). Suppose that the label of $v$ is $a$. If $a \in B \cup \Lab$, we can
finish processing $v$. Assume that $a \notin B \cup \Lab$. Pick a label $b
\in B$ that does not appear in the processed descendants of $u$, nor in $u
$ itself. We can always find such a label $b$ because the number of
processed descendants of $u$ (including $u$ itself) is bounded by
$\sum_{i=0}^{n-1} R^i = \frac{R^n-1}{R-1} <
R^n\leq R^{|\Ppr|}$, and so is the number of labels from $B$
used in these nodes. Let $c \in \Sigma \setminus (B\cup \Lab)$ be a
fresh label. We now replace all appearances of $b$ with $c$, but only
in the unprocessed descendants of the node $u$. Observe that these nodes are
separated from the nodes that keep their label $b$ by distance at least
$n$. Next, we replace all appearances of $a$ with $b$, but only in the
unprocessed descendants of $u$. Again, the distance from these nodes
to the other nodes with label $a$ or $b$ is at least $n$.  Thus, the
modification does not affect the outcome of any label comparison done
by rules in $\Ppr$ (because they use only the short
axis and are connected). After all nodes are processed, all labels in $t'$ are from
$B \cup \Lab$. 
\end{proof}
Let $\Sigma_0$ be the finite alphabet from the previous Lemma.
Now we can construct an automaton $W_{\Ppr}$, recognizing the set of witnesses for $\Ppr$.
From Lemma \ref{lem:trees:automata} we get a two-way alternating tree automaton $\Bb_{\Ppr}$ which 
works over $\Sigma_0 \times \{0,1\}$,
and accepts the set of trees that have only one node labeled with $(a,1)$ for $a \in \Sigma_0$, and the goal predicate of $\Ppr$ is satisfied in this node.
The size of this automaton is exponential in $|\Ppr|$.
Let $\Aa_{\Ppr}$ be the bottom-up automaton recognizing $L(\Bb_{\Ppr})$ obtained via Proposition~\ref{prop:two-way}.
Let ${\cal N}_{\Ppr}$ be an automaton obtained by taking a product of the bottom-up automaton recognizing the complement of $L(\Bb_{\Ppr})$ (again obtained via Proposition~\ref{prop:two-way}) and 
the automaton checking that there is only one node in the tree with label $(a,1)$ for some $a \in \Sigma_0$.
Then ${\cal N}_{\Ppr}$ accepts all trees over $\Sigma_0$ for which $\Ppr$ does not hold in the marked node.
The size of both $\Aa_{\Ppr}$ and ${\cal N}_{\Ppr}$ is double exponential in $|\Ppr|$.

With those two automata, the construction of $W_{\Ppr}$ is easy. The set of states of $W_{\Ppr}$ is 
\begin{equation*}
Q(\Aa_{\Ppr}) \times \left( \{ \epsilon , \mathsf{OK} \} \cup Q({\cal N}_{\Ppr}) \right) 
\end{equation*}
where $Q(A)$ denotes the set of states of the automaton $A$.
Let $t$ be a tree over $\Sigma_0 \times \{0,1\}$ and let $X$ denote the marked node. 
The automaton $W_{\Ppr}$ starts in the state $(q_I, \epsilon)$, where $q_I$ is the initial state of $\Aa_{\Ppr}$. 
Then $W_{\Ppr}$ simulates $\Aa_{\Ppr}$ on $t$. 
In any node of a tree, the automaton $W_{\Ppr}$ can guess that here begins the neighbourhood of $X$ in which $\Ppr$ does not hold. 
Then $W_{\Ppr}$ changes the second component of its state from $\epsilon$ to the initial state of ${\cal N}_{\Ppr}$ and simulates ${\cal N}_{\Ppr}$ on the guessed neighbourhood, verifying that indeed $\Ppr$ does not hold in it. If $W_{\Ppr}$ has reached an accepting state of ${\cal N}_{\Ppr}$, it can guess that this node is the root of the neighbourhood and change the state to $\mathsf{OK}$ in the second component. Accepting states of $W_{\Ppr}$ are states $(q, \mathsf{OK})$ where $q$ is any accepting state of $\Aa_{\Ppr}$.

Similarly to the word case, if there exists a witness of size linear in the size of the automaton $W_{\Ppr}$, then there exist arbitrarily big witnesses.
\begin{lemma}
Let $N$ be the number of states of the automaton $W_{\Ppr}$. If there exists a $(2N+2)$-witness for $\Ppr$, 
then there exist $n$-witnesses for arbitrary large $n$. The existence of $(2N+2)$-witness can be decided in time polynomial in $N$.
\end{lemma}
\begin{proof}
We use a very similar pumping argument as in the word case. 
This time, however, to obtain arbitrarily big witnesses we need to 
be able to pump every path of the neighbourhood in which $\Ppr$ is not satisfied. 
%This is why we need $2N+2$ instead of $N$-witnesses, as in the word's case.

Suppose that there exists a $(2N+2)$-witness and let $X$ be the marked node.
Then on every path of length $2N+2$ from $X$ downwards, some state of $W_{\Ppr}$ must repeat, so we can pump the context between those nodes. Notice that some paths may be shorter, because the $(2N+2)$-witness may contain a leaf of the tree -- we don't need to pump those paths.
On the path from $X$ upwards of length $N+1$ again some states of $W_{\Ppr}$ repeat, and we can pump the context between the occurrences of the same state.
This time, however, we need also to extend the paths that start on the pumped fragment and go downwards, but do not return to $X$. 
Every such path is of length at least $N+1$ (that is why we need the $2N+2$ size of the neighbourhood), so we can pump each of them (except for those that are shorter because they end with a leaf of the tree). 

To verify the existence of a $(2N+2)$-witness
we modify the automaton $W_{\Ppr}$ by adding two counters from $0$ to $2N+2$. 
When the automaton guesses the beginning of a neighbourhood of $X$ in a non-leaf node $Y$
it starts counting the length of the shortest path until the least common ancestor of $Y$ and $X$ is reached.
The automaton in a node calculates the length of the shortest path as $1$ + the minimum of the values of the counters calculated for its children (if the value of the counter is $2N+2$, adding $1$ does not change its value).
When a neighbourhood of $X$ begins in a leaf of the tree, the length of this path does not need to be $2N+2$, so the automaton sets the counter to $2N+2$ (that is -- sufficient length).
The second counter is used only for the nodes on the path above $X$ and counts the length of the path for $X$ to this node (for any other node in the guessed neighbourhood, value of this counter is 0).

It is not difficult to see that using those two counters we can come up with an acceptance condition such that the modified automaton has an accepting run iff there exists a $(2N+2)$-witness for~$\Ppr$.
Since emptiness can be decided in time linear in the size of the automaton, we get the claim. \qed
\end{proof}
%To verify the existence of a $(2N+2)$-witness,
%we modify the automaton $W_{\Ppr}$ by adding a counter from $0$ to $2N+2$. 
%When the automaton guesses the beginning of a neighbourhood of $X$ in a non-leaf node 
%it starts counting the length of the shortest path until $X$ or the root of the neighbourhood is reached. 
%The automaton in a node calculates the length of the shortest path as 1 + a minimum of the values of the counter calculated for the children. 
%Of course if the value of the counter is $2N+2$, adding 1 to it does not change the value of the counter.
%Additionally, the automaton resets the counter in $X$ and starts counting the length of the path from $X$ to the root of the guessed neighbourhood. 
%When a neighbourhood of $X$ begins in a leaf of the tree, the length of this path does not need to be $2N+2$, so the automaton sets the counter to $2N+2$ (that is, sufficient length).
%
%The automaton accepts if every path (that does not start in a leaf) leading upwards to $X$
%has length at least $2N+2$ and every path that leads to the root of the guessed neighbourhood has length at least $N+1$.
%It is clear that this modified automaton has an accepting run iff there exists a $(2N+2)$-witness for~$\Ppr$.
%Since emptiness can be decided in time linear in the size of the automaton, we get the claim. \qed
%\end{proof}
Since the size of $W_{\Ppr}$ is double exponential in $|\Ppr|$, we get a \TwoExpTime{} procedure for deciding boundedness of $\Ppr$. \qed
\end{proof}

\end{document}